\documentclass[a4paper,UKenglish,cleveref, autoref, thm-restate, svgnames]{lipics-v2021}

\pdfoutput=1 
\hideLIPIcs  

\title{Spartan Bipartite Graphs are Essentially Elementary} 

\titlerunning{Spartan Bipartite Graphs are Essentially Elementary} 

\author{Neeldhara Misra}{Department of Computer Science and Engineering,\and Indian Institute of Technology, Gandhinagar \and \url{http://www.neeldhara.com} }{neeldhara.m@iitgn.ac.in}{https://orcid.org/0000-0003-1727-5388}{}

\author{Saraswati Girish Nanoti}{Department of Mathematics,\and Indian Institute of Technology, Gandhinagar}{nanoti\_saraswati@iitgn.ac.in}{}{}

\authorrunning{N. Misra and S. Nanoti} 

\Copyright{N. Misra and S. Nanoti} 



\category{} 

\relatedversion{A shorter version of this paper will appear in the Proceedings of the 48th International Symposium on Mathematical Foundations of Computer Science (MFCS 2023).} 

\funding{Both authors acknowledge support from IIT Gandhinagar. Further, the first author acknowledges support from the SERB grant ECR/2018/002967 and the second author acknowledges support from CSIR.}

\acknowledgements{The authors are grateful to several anonymous reviewers for helpful comments on earlier versions of this manuscript.}

\nolinenumbers 

\EventEditors{J\'{e}r\^{o}me Leroux, Sylvain Lombardy, and David Peleg}
\EventNoEds{3}
\EventLongTitle{48th International Symposium on Mathematical Foundations of Computer Science (MFCS 2023)}
\EventShortTitle{MFCS 2023}
\EventAcronym{MFCS}
\EventYear{2023}
\EventDate{August 28 to September 1, 2023}
\EventLocation{Bordeaux, France}
\EventLogo{}
\SeriesVolume{272}
\ArticleNo{68}

\usepackage{charter,parskip}
\usepackage{caption}
\usepackage{subcaption}
\usepackage{comment}
\usepackage{appendix}

\usepackage[T1]{fontenc}
\usepackage{graphicx}
\usepackage[export]{adjustbox}
\usepackage{float}
\usepackage{smartdiagram}
\usepackage{etoolbox}
\makeatletter
\patchcmd{\BR@backref}{\newblock}{\newblock($\uparrow$~}{}{}
\patchcmd{\BR@backref}{\par}{)\par}{}{}
\makeatother
\hypersetup{
    colorlinks=true,
    linkcolor=IndianRed,
    filecolor=magenta,      
    urlcolor=DodgerBlue,
    citecolor = SeaGreen
}
\usepackage{caption}
\usepackage{subcaption}
\usepackage[backgroundcolor=SkyBlue,linecolor=SkyBlue,obeyFinal]{todonotes}

\usepackage{wrapfig}
\makeatletter
\providecommand\@dotsep{5}
\def\listtodoname{}
\def\listoftodos{\@starttoc{tdo}\listtodoname}
\makeatother

\newcounter{nmcomment}

\usepackage{comment}
\excludecomment{shortver}
\includecomment{longver}

\usepackage{wrapfig}
\usepackage[font=small]{caption, subcaption}
\usepackage{tikz}
\usetikzlibrary{positioning}
\usetikzlibrary{arrows}
\usetikzlibrary{decorations}
\usetikzlibrary{decorations.pathmorphing}
\usetikzlibrary{decorations.markings}
\usetikzlibrary{decorations.pathreplacing}
\usetikzlibrary{shapes.geometric}
\tikzset{
    small circles/.style={circle,inner sep=2pt,fill=#1},
    hollow circles/.style n args={2}{circle,inner sep=#1,draw=#2,thick},
    stars/.style={star,inner sep=2pt}
}
    
\usepackage{latexsym,amssymb,amsmath,mathtools,eulervm,xspace}

\newcommand{\NPH}{\ensuremath{\mathsf{NP}}-hard\xspace}

\newcommand{\EVC}{\textsc{Eternal Vertex Cover}}

\usepackage{thmtools, thm-restate}

\begin{document}

\maketitle

\begin{abstract}
    We study a two-player game on a graph between an attacker and a defender. To begin with, the defender places guards on a subset of vertices. In each move, the attacker attacks an edge. The defender must move at least one guard across the attacked edge to defend the attack. The defender wins if and only if the defender can defend an infinite sequence of attacks. The smallest number of guards with which the defender has a winning strategy is called the eternal vertex cover number of a graph $G$ and is denoted by $evc(G)$. It is clear that $evc(G)$ is \emph{at least} $mvc(G)$, the size of a minimum vertex cover of $G$. We say that $G$ is \emph{Spartan} if $evc(G) = mvc(G)$. The characterization of Spartan graphs has been largely open. In the setting of bipartite graphs on $2n$ vertices where every edge belongs to a perfect matching, an easy strategy is to have $n$ guards that always move along perfect matchings in response to attacks. We show that these are essentially the only Spartan bipartite graphs.
\end{abstract}

\section{Introduction}
\label{sec:introduction}

A \emph{vertex cover} of a graph $G$ is a subset $S$ of vertices such that every edge in $G$ has at least one of its endpoints in $S$. An \emph{optimal} vertex cover of a graph $G$ is a vertex cover of the smallest possible size and the size of this optimal vertex cover is denoted by $mvc(G)$. A \emph{bipartite graph} $G = (V,E)$ is a graph whose vertex set can be partitioned into two independent sets, say $V = (A \cup B)$, that is every edge is between a vertex in $A$ and one in $B$. Clearly, both $A$ and $B$ are vertex covers of $G$. If a bipartite graph $G = (A \cup B, E)$ is connected and its \emph{only} optimal vertex covers are $A$ and $B$, then we say that $G$ is \emph{elementary}. If every connected component of a bipartite graph is elementary, then we call it \emph{essentially elementary}.

Klostermeyer and Mynhardt~\cite{KM09} introduced the notion of the \emph{eternal vertex cover} number of a graph $G$. There are two players --- an \emph{attacker} and a \emph{defender}, who are playing on a simple, undirected graph $G$. In the beginning, the defender can choose to place guards on some of the vertices of $G$. The attacker's move involves choosing an edge to ``attack''. The defender is able to ``defend'' this attack if she can move the guards along the edges of the graph in such a way that at least one guard moves along the attacked edge. If both the endpoints of the attacked edge have guards on them, the guards can simply exchange their places and the attack is defended as well as the original configuration of guards is restored. Note that any number of guards can move after an attack, but each guard can move only one ``step'', i.e., just along a single edge. If such a movement is not possible, then the attacker wins. If the defender can defend the graph against an infinite sequence of attacks, then the defender wins (see~\Cref{fig:intro}). The minimum number of guards with which the defender has a winning strategy is called the \emph{eternal vertex cover number} of the graph $G$ and is denoted by $evc(G)$. 

There are two variations of this game, one where only one guard can be present per vertex at a given time and a variant where more than one guard can be present on a vertex at a given time. Our results apply to both the variants of the game.

\begin{figure}[ht]
\centering
\begin{subfigure}[b]{0.42\textwidth}
\resizebox{\textwidth}{!}{%
\begin{tikzpicture}[scale=0.25]
\node[small circles=DodgerBlue] (A1) at (0,0) {};
\node[small circles=DodgerBlue] (B1) [below of=A1] {};

\foreach \x in {2,3,...,6}{
\pgfmathtruncatemacro{\y}{\x-1}
\node[small circles=DodgerBlue] (A\x) [right of=A\y] {};
}
\foreach \x in {2,3,...,5}{
\pgfmathtruncatemacro{\y}{\x-1}
\node[small circles=DodgerBlue] (B\x) [right of=B\y] {};
};
\node (B6) [right of = B5] {};    
\node[small circles=DodgerBlue] (B7) [right of = B6] {};

\draw (A1) -- (B1) -- (B2) -- (A2) -- (B3) -- (A3) -- (A4) -- (B4)  -- (B5) -- (A4) -- (A5) -- (A6) -- (B5) -- (A6) -- (B7);
\draw (A1) -- (A2);
\draw (B1) -- (A2);
\draw (A1) -- (B2);

\node[star,fill=DarkOrange,inner sep=2pt] at (A1) {};
\node[star,fill=DarkOrange,inner sep=2pt] at (A2) {};
\node[star,fill=DarkOrange,inner sep=2pt] at (B1) {};
\node[star,fill=DarkOrange,inner sep=2pt] at (A3) {};
\node[star,fill=DarkOrange,inner sep=2pt] at (A4) {};
\node[star,fill=DarkOrange,inner sep=2pt] at (B5) {};
\node[star,fill=DarkOrange,inner sep=2pt] at (A6) {};
\node[star,fill=DarkOrange,inner sep=2pt] at (B4) {};
\end{tikzpicture}
}
\caption{The intial positions of the guards are denoted by the star-shaped vertices.}
\label{fig:intro1} 
\end{subfigure}
\begin{subfigure}[b]{0.42\textwidth}
\resizebox{\textwidth}{!}{%
\begin{tikzpicture}[scale=0.25]


\node[small circles=DodgerBlue] (A1) at (0,0) {};
\node[small circles=DodgerBlue] (B1) [below of=A1] {};

\foreach \x in {2,3,...,6}{
\pgfmathtruncatemacro{\y}{\x-1}
\node[small circles=DodgerBlue] (A\x) [right of=A\y] {};
}
\foreach \x in {2,3,...,5}{
\pgfmathtruncatemacro{\y}{\x-1}
\node[small circles=DodgerBlue] (B\x) [right of=B\y] {};
};
\node (B6) [right of = B5] {};    
\node[small circles=DodgerBlue] (B7) [right of = B6] {};

\draw (A1) -- (B1) -- (B2) -- (A2) -- (B3) -- (A3) -- (A4) -- (B4)  -- (B5) -- (A4) -- (A5) -- (A6) -- (B5) -- (A6) -- (B7);
\draw (A1) -- (A2);
\draw (B1) -- (A2);
\draw (A1) -- (B2);

\node[star,fill=DarkOrange,inner sep=2pt] at (A1) {};
\node[star,fill=DarkOrange,inner sep=2pt] at (A2) {};
\node[star,fill=DarkOrange,inner sep=2pt] at (B1) {};
\node[star,fill=DarkOrange,inner sep=2pt] at (A3) {};
\node[star,fill=DarkOrange,inner sep=2pt] at (A4) {};
\node[star,fill=DarkOrange,inner sep=2pt] at (B5) {};
\node[star,fill=DarkOrange,inner sep=2pt] at (A6) {};
\node[star,fill=DarkOrange,inner sep=2pt] at (B4) {};
\draw[thick,color=IndianRed,decoration = {zigzag,segment length = 3pt,amplitude=1pt},decorate] (A6) -- (B7);
\end{tikzpicture}
}
\caption{The attacker's move targets the edge to the far-right, highlighted by a wavy red line.}
\label{fig:intro2}
\end{subfigure}
\begin{subfigure}[b]{0.42\textwidth}
\resizebox{\textwidth}{!}{%
\begin{tikzpicture}


\node[small circles=DodgerBlue] (A1) at (0,0) {};
\node[small circles=DodgerBlue] (B1) [below of=A1] {};

\foreach \x in {2,3,...,6}{
\pgfmathtruncatemacro{\y}{\x-1}
\node[small circles=DodgerBlue] (A\x) [right of=A\y] {};
}
\foreach \x in {2,3,...,5}{
\pgfmathtruncatemacro{\y}{\x-1}
\node[small circles=DodgerBlue] (B\x) [right of=B\y] {};
};
\node (B6) [right of = B5] {};    
\node[small circles=DodgerBlue] (B7) [right of = B6] {};

\draw (A1) -- (B1) -- (B2) -- (A2) -- (B3) -- (A3) -- (A4) -- (B4)  -- (B5) -- (A4) -- (A5) -- (A6) -- (B5) -- (A6) -- (B7);
\draw (A1) -- (A2);
\draw (B1) -- (A2);
\draw (A1) -- (B2);

\node[star,fill=DarkOrange,inner sep=2pt] at (A1) {};
\node[star,fill=DarkOrange,inner sep=2pt] at (A2) {};
\node[star,fill=DarkOrange,inner sep=2pt] at (B1) {};
\node[star,fill=DarkOrange,inner sep=2pt] at (A3) {};
\node[star,fill=DarkOrange,inner sep=2pt] at (A4) {};
\node[star,fill=DarkOrange,inner sep=2pt] at (B5) {};
\node[star,fill=DarkOrange,inner sep=2pt] at (B7) {};
\node[star,fill=DarkOrange,inner sep=2pt] at (B4) {};
\draw[thick,color=SeaGreen,decoration = {zigzag,segment length = 3pt,amplitude=1pt},decorate] (A6) -- (B7);

\end{tikzpicture}
}
\caption{The defender responds to defend the attack by moving a guard along the attacked edge.}
\label{fig:Ng2}
\end{subfigure}
\begin{subfigure}[b]{0.42\textwidth}
\resizebox{\textwidth}{!}{%
\begin{tikzpicture}


\node[small circles=DodgerBlue] (A1) at (0,0) {};
\node[small circles=DodgerBlue] (B1) [below of=A1] {};

\foreach \x in {2,3,...,6}{
\pgfmathtruncatemacro{\y}{\x-1}
\node[small circles=DodgerBlue] (A\x) [right of=A\y] {};
}
\foreach \x in {2,3,...,5}{
\pgfmathtruncatemacro{\y}{\x-1}
\node[small circles=DodgerBlue] (B\x) [right of=B\y] {};
};
\node (B6) [right of = B5] {};    
\node[small circles=DodgerBlue] (B7) [right of = B6] {};

\draw (A1) -- (B1) -- (B2) -- (A2) -- (B3) -- (A3) -- (A4) -- (B4)  -- (B5) -- (A4) -- (A5) -- (A6) -- (B5) -- (A6) -- (B7);
\draw (A1) -- (A2);
\draw (B1) -- (A2);
\draw (A1) -- (B2);

\node[star,fill=DarkOrange,inner sep=2pt] at (A1) {};
\node[star,fill=DarkOrange,inner sep=2pt] at (A2) {};
\node[star,fill=DarkOrange,inner sep=2pt] at (B1) {};
\node[star,fill=DarkOrange,inner sep=2pt] at (A3) {};
\node[star,fill=DarkOrange,inner sep=2pt] at (A5) {};
\node[star,fill=DarkOrange,inner sep=2pt] at (B5) {};
\node[star,fill=DarkOrange,inner sep=2pt] at (B7) {};
\node[star,fill=DarkOrange,inner sep=2pt] at (B4) {};
\end{tikzpicture}
}
\caption{Simultaneously, the defender moves another guard to ensure that no edges are left vulnerable. This is the resultant position of the guards.}
\label{fig:intro4}
\end{subfigure}
\caption{An attack that is defended by moving two guards.}
\label{fig:intro}
\end{figure}
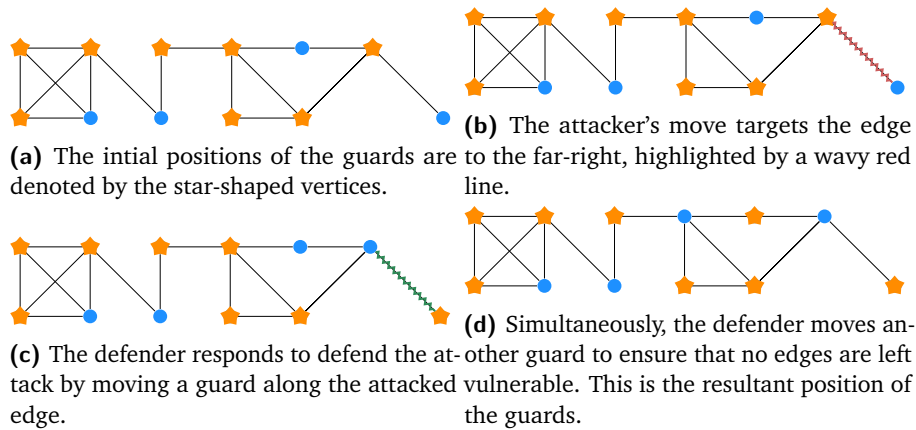

If $S_\ell$ is the subset of vertices that have guards on them after the defender has played her $\ell$-th move, and $S_\ell$ is not a vertex cover of $G$, then the attacker can target any of the uncovered edges to win the game. Therefore, when the defender has a winning strategy, it implies that she can always ``reconfigure'' one vertex cover into another in response to any attack, where the reconfiguration is constrained by the rules of how the guards can move and the requirement that at least one of these guards needs to move along the attacked edge. It follows that $evc(G)\geq mvc(G)$, where $mvc(G)$ is the size of the smallest vertex cover for any graph $G$.

We call graphs $G$ which enjoy $evc(G) = mvc(G)$ \emph{Spartan}, to indicate that these graphs can manage attacks without using ``additional'' guards. Understanding Spartan graphs for several special cases have been addressed in the literature~(see, for instance, \cite{BCFPRW2021}). However, the characterization of bipartite graphs that are Spartan has been left open. We show that bipartite graphs are Spartan if and only if they are essentially elementary.




We also consider a natural extension of the eternal vertex cover problem: where the guards are allowed to move two or more steps on each turn instead of one, without retracing. An attack is defended only if the guard moves along the attacked edge on their first step. We show that with just one additional step, we need either just as many guards as the size of the smallest vertex cover or one more: so in this variant all graphs are ``almost'' Spartan. We also show that the extra guard is needed only if the graph has a degree one vertex not contained in any minimum vertex cover. Finally, we show that allowing for more than one extra step is the same as allowing exactly one extra step. 

\textbf{Related Work.} Among characterizations related to Spartan graphs, we have the following. Let the graph class $\mathcal{F}$ denote the class of all connected graphs $G$ for which each minimum vertex cover of $G$ that contains all the cut vertices of $G$ induces a connected subgraph in $G$. (A cut vertex is a vertex whose removal disconnects the graph.) Let $G(V,E)$ be a graph that belongs to $\mathcal{F}$, with at least two vertices, and $X\subset V$ be the set of cut vertices of $G$. Then Babu et al.~\cite{BCFPRW2019} showed that $G$ is Spartan if and only if for every vertex $v\in V\backslash X$, there exists a minimum vertex cover $S_v$ of $G$ such that $X\cup \{v\}\subset S_v$. Klostermeyer and Mynhardt~\cite{KM2011} also study graphs for which the eternal vertex cover number coincides with the eternal domination number, a closely related notion. This is a similar game, except that the attacks happen on vertices and if the attacked vertex does not have a guard already, then a guard from a neighbouring vertex must come to the attacked vertex. All other guards can stay on their initial position or move to a neighbouring vertex. The minimum number of guards required to protect the graph from an infinite sequence of attacks is called the \emph{eternal domination number}. 

Note that twice as many vertices as the $mvc(G)$ always suffice to defend against any sequence of attacks --- by placing guards on both endpoints of any maximum matching to begin with and after any attack, reconfiguring the guards to obtain another maximum matching. Using this strategy, a  $2-$approximation algorithm for $\EVC$ was obtained  by Fomin et al.~\cite{FGGKS2010}. This also implies $mvc(G)\leq evc(G) \leq 2mvc(G)$. \cite{KM09} gave a characterization of the graphs for which the upper bound is achieved. The notion of elementary graphs was considered by Lov\'asz and Plummer~\cite{LLPM-matchingbook} and several useful characterizations were given by Hetyei~\cite{hetyei1964rectangular}. We adapt these definitions to suit the context of bipartite graphs.


\paragraph*{Methodology.}

Recall that a \emph{matching} is a collection of vertex disjoint edges, and a matching which contains one edge incident to each vertex of a graph is called a \emph {perfect matching}. We say that an edge is \emph{allowed} if it is contained in some perfect matching. A graph $G$ is said to be \emph{elementary} if and only if it is connected and every edge is allowed.

It turns out that a connected bipartite graph $G = (A \cup B, E)$ is elementary if and only if its \emph{only} optimal vertex covers are $A$ and $B$, as we make explicit in the proposition below. 


\begin{proposition}[\cite{hetyei1964rectangular}]\label{prop1}
The following are equivalent for any connected bipartite graph $G = (A \cup B, E)$.
\begin{enumerate}
\item $A$ and $B$ are the only minimum vertex covers of $G$, and in particular $|A| = |B|.$
\item Every edge in $G$ is allowed.
\end{enumerate}
\end{proposition}


Notice that it is easy to see that for bipartite graphs, if $G$ is elementary then it is Spartan (see also~\Cref{ele} and Araki, Fujito, and Inoue~\cite{AFI2015}). Indeed, we start by placing guards on all vertices of $A$. If the edge $e$ is attacked, then we move guards to $B$ along the edges of the perfect matching $M$ that contains the edge $e$. Future attacks can be similarly defended, so the guards alternate between occupying $A$ and $B$ in response to edge attacks.  

What we focus on demonstrating in this contribution is the converse, namely that if a connected bipartite graph $G$ does not have this structure, then it is also not Spartan: in some sense, the trivial scenario is the only one in which we can hope to manage without employing any ``extra'' guards. 

\textbf{Spartan $\implies$ Perfect Matching.} To begin with, notice that if $G = (A \cup B, E)$ is connected and does not have a perfect matching, we already need more than $mvc(G)$ guards (see also~\Cref{prop:spartanimpliespm}). In particular, let $M$ be a maximum matching in $G$. Recall that $|M| = mvc(G)$ by K\"{o}nig's theorem. Without loss of generality, let $b \in B$ be a vertex not incident to any edge of $M$. Since $G$ is connected, $b$ has some neighbor $a \in A$, which must belong to any vertex cover of $G$. In particular, the defender is forced to position a guard on $a$ in the initial configuration. If the edge $ab$ is attacked, the guard on $a$ is forced to move to $b$. This creates a situation where $|M|-1$ guards have to reposition themselves to protect all the edges of $M$ --- which is impossible since the edges of $M$ are disjoint. 

Therefore, a necessary condition for a connected bipartite graph to be Spartan is that it must admit a perfect matching. This is, however, evidently not sufficient: indeed, there are connected bipartite graphs with perfect matchings that are not Spartan (see also~\Cref{fig:intro3}).

\begin{figure}[ht]
\centering
\resizebox{0.2\textwidth}{!}{%
    \begin{tikzpicture}[scale=0.23] {\tikzset{decoration={snake,amplitude=.4mm,segment length=2mm, post length=0mm,pre length=0mm}}
\node[small circles=DodgerBlue] (A1) at (0,0) {};
\node[small circles=DodgerBlue] (A2) [below of=A1] {};
\node[small circles=DodgerBlue] (B1) [right of=A1] {};
\node[small circles=DodgerBlue] (B2) [right of=A2] {};

\draw (A1) -- (B1);
\draw[red,decorate] (B1) -- (A2);
\draw (A2) -- (B2);

\node[star,fill=DarkOrange,inner sep=2pt] at (B1) {};
\node[star,fill=DarkOrange,inner sep=2pt] at (B2) {};
}
\end{tikzpicture}
}
~~~~~~~~
\resizebox{0.2\textwidth}{!}{%
    \begin{tikzpicture}[scale=0.25]
    {\tikzset{decoration={snake,amplitude=.4mm,segment length=2mm, post length=0mm,pre length=0mm}}
\node[small circles=DodgerBlue] (A1) at (0,0) {};
\node[small circles=DodgerBlue] (A2) [below of=A1] {};
\node[small circles=DodgerBlue] (B1) [right of=A1] {};
\node[small circles=DodgerBlue] (B2) [right of=A2] {};

\draw[red,decorate] (A1) -- (B1);
\draw (B1) -- (A2);
\draw (A2) -- (B2);

\node[star,fill=DarkOrange,inner sep=2pt] at (B1) {};
\node[star,fill=DarkOrange,inner sep=2pt] at (A2) {};
}
\end{tikzpicture}
}
~~~~~~~~
\resizebox{0.2\textwidth}{!}{%
    \begin{tikzpicture}[scale=0.25]{\tikzset{decoration={snake,amplitude=.4mm,segment length=2mm, post length=0mm,pre length=0mm}}
\node[small circles=DodgerBlue] (A1) at (0,0) {};
\node[small circles=DodgerBlue] (A2) [below of=A1] {};
\node[small circles=DodgerBlue] (B1) [right of=A1] {};
\node[small circles=DodgerBlue] (B2) [right of=A2] {};

\node[star,fill=DarkOrange,inner sep=2pt] at (A1) {};
\node[star,fill=DarkOrange,inner sep=2pt] at (A2) {};

\draw (A1) -- (B1);
\draw[red,decorate] (B1) -- (A2);
\draw (A2) -- (B2);
}
\end{tikzpicture}
}
\caption{A graph with a perfect matching where $mvc(G)$ many guards do not suffice. For each choice of $mvc(G)$ many guards at the initial configuration, attacking the dashed edge shown in the figure leads to an indefensible position in one or two steps.}
\label{fig:intro3}
\end{figure}
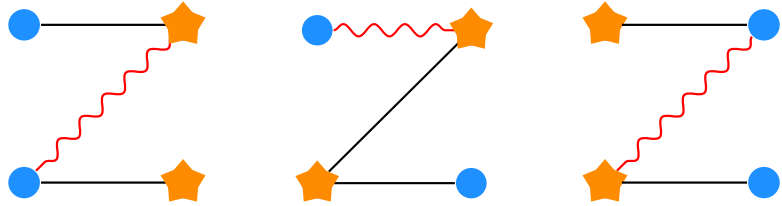

\textbf{Spartan $\implies$ No Degree One Vertices.} We observe that a connected bipartite Spartan graph $G = (A \cup B, E)$ on at least three vertices cannot have a degree one vertex (see also~\Cref{prop:nodeg1} and Babu et al.~\cite{BCFPRW2019}).  Assume to the contrary that $G$ has a degree one vertex, say $a \in A$. Let $b$ denote $a$'s unique neighbor in $G$. Since $G$ has at least one vertex other than $a$ and $b$, and $G$ is connected, $b$ must have another neighbor in $G$: say $a^\prime$. Now, if the initial configuration has a guard on $a$, then attack the edge $ba^\prime$; otherwise attack the edge $ab$ to force a guard on $a$ and then attack $ba^\prime$. Both cases lead to a scenario where both endpoints of the edge $ab$ have guards on them. However, since $G$ is Spartan, it has a perfect matching $M$. Further, since $a$ is a degree one vertex, $M$ is forced to contain the edge $ab$. This in turn implies that at least $|A|-1$ guards are required to defend the rest of the graph, contradicting the assumption that $G$ is Spartan.

\textbf{Proof by Structure-Preserving Contractions.} So far, we have seen that a connected bipartite graph $G = (A \cup B, E)$ which is Spartan must have a perfect matching and no degree one vertices. Fix an arbitrary perfect matching $M$, and let $ab \in M$. Let $P$ be a maximal walk (without repeating edges) starting at $a$, alternating between edges of $M$ and $E \setminus M$. The terminal vertex of $P$ is either $a$, or a vertex $v$ which must have a neighbor on $P$, since $G$ has no degree one vertices. So from $P$, we can derive a cycle whose edges alternate between $M$ and $E \setminus M$ (see also~\Cref{NAC}).

At a high level (see~\Cref{fig:my_label-a}), our proof relies on ``contracting'' each such cycle to a single matching edge and arguing that the resulting graph remains Spartan provided the original graph was also Spartan. In particular, we show the following. If $G$ is Spartan to begin with, and $M$ is a perfect matching with $S$ is a set of endpoints of $M'\subseteq M$ such that $G[S]$ is elementary, then ``replacing'' $G[S]$ with a single edge until this operation is no longer possible; keeps the graph Spartan (see also~\Cref{MCP}). Here it is useful that any bipartite graph obtained by adding edges to a cycle (while preserving the bipartite-ness) is elementary, a fact that we establish separately (see also~\Cref{cycle}) --- thus we can simply keep contracting along cycles $C$. The idea of the proof is then the following. Consider what happens once we are stuck, i.e., we cannot contract any further: we either have an edge, or we have a graph that is not an edge but must have a degree one vertex (if not, note that we would not be stuck). But from our previous discussion, we know that a connected bipartite Spartan graph on at least three vertices cannot have a degree one vertex. So in this situation, we have transformed the Spartan graph we started with into one that is not Spartan, and in particular, there is an attack that $mvc(\cdot)$ many guards will not be able to defend on the ``contracted'' graph. It turns out that this attack can be mimicked on the graph we started with, contradicting our assumption that the original graph is Spartan. Thus, our process must end at a single edge.


To conclude the proof, we show that if we take an elementary bipartite graph $G$ and ``inject'' an elementary graph into it by substituting an edge with an elementary graph, then the resulting graph is also elementary. This lets us run the contraction operations on $G$ in reverse, starting from the edge that we ended up with; and since an edge is elementary, we conclude that $G$ must have also been elementary (see also~\Cref{MC}). 

\begin{figure}
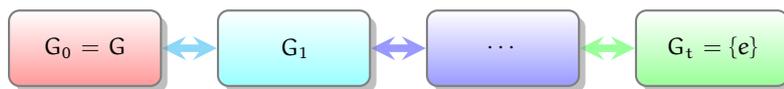

    \centering
    \smartdiagramset{arrow style=<->,back arrow disabled=true}
\smartdiagram[flow diagram:horizontal]{$G_0 = G$,$G_1$,$\cdots$,$G_t = \{e\}$}
    \caption{Proof idea: The sequence of forward contractions preserves ``Spartan-ness'' overall (i.e, $G_t$ is Spartan if $G_0$ is Spartan) while the backward expansions preserve ``elementary-ness''.}
    \label{fig:my_label-a}
\end{figure}

When $G$ is not connected, the arguments we made can be made independently of every connected component of $G$. This leads to our final characterization: that a biparitite graph $G$ is Spartan if and only if each of its components are elementary. Since it is straightforward to check if a bipartite graph $G = (A \cup B, E)$ is elementary (for instance, by checking if $G \setminus \{u,v\}$ admits a perfect matching for every $uv \in E$), our result also implies a polynomial-time recognition algorithm for bipartite Spartan graphs. We note that, in contrast to $mvc(G)$, determining $evc(G)$ is known to be NP-hard even for bipartite graphs.

In~\Cref{sec:wip}, we describe the proofs of the claims mentioned in the outline here in greater detail. Our notation largely follows conventions from Diestel~\cite{graphsbook}. 

\section{Spartan Bipartite Graphs are Essentially Elementary}
\label{sec:wip}

We begin with the (easy) forward implication, which was also observed by~\cite{AFI2015} --- we provide an explicit proof here in the interest of completeness.

\begin{lemma}\label{ele}
If a bipartite graph $G$ is essentially elementary, $evc(G)=|V(G)|/2=mvc(G)$, i.e., $G$ is Spartan.
\end{lemma}

\begin{proof}
We first make the argument for the case when $G$ is connected. Since $G$ is elementary, $mvc(G)=|V(G)|/2$ by the previous proposition. Since $evc(G)\geq mvc(G)$ for any graph $G$, we just need to show that the defender has a winning strategy when the number of guards is $mvc(G)$. Initially place all the guards on one side, wlog say $A$. Suppose the attacker attacks an edge $e$, since $e$ is allowed there exists a perfect matching containing $e$. Use this matching to transfer all the guards to $B$ to ensure that one guard has moved along $e$. Similarly transfer all the guards to $A$ using the perfect matching containing the attacked edge in the next attack. Since the defender can always transfer all the guards to the other side ensuring the movement of one guard across the attacked edge, any sequence of attacks can be defended and thus the defender wins.

Now, if $G$ has more than one connected component, the argument above can be repeated for each connected component. If the components have sizes $|C_1|, \ldots, |C_\ell|$ then the overall size of the eternal vertex cover will be: $\frac{1}{2}\sum_{i=1}^\ell |C_i| = |V(G)|/2$. 
\end{proof}

Now, let $G$ be a Spartan biparitite graph with connected components $C_1, \ldots, C_\ell$. Note that: 

$$\mathsf{evc}(G) = \sum_{i=1}^{\ell}\mathsf{evc}(G[C_i]) \geqslant \sum_{i=1}^{\ell}\mathsf{mvc}(G[C_i]) = \mathsf{mvc}(G).$$


Therefore, if $G$ is Spartan then $G[C_i]$ is Spartan for all $1 \leqslant i \leqslant [\ell]$. Indeed, if not, then there exists a component $C_i$ for which $\mathsf{evc}(G[C_i]) > \mathsf{mvc}(G[C_i])$. But combined with the inequality above, this will imply that $\mathsf{evc}(G) > \mathsf{mvc}(G)$, contradicting our assumption that $G$ is Spartan. Therefore, we may assume without loss of generality that $G$ is connected, and our goal amounts to showing the following.


\begin{lemma}\label{conv}
If $G$ is a connected bipartite graph that is Spartan, then it is elementary. 
\end{lemma}

We first argue that any connected Spartan bipartite graph must have a perfect matching. 

\begin{proposition}
\label{prop:spartanimpliespm}
If $G$ is a Spartan connected bipartite graph, it admits a perfect matching.
\end{proposition}

\begin{proof}
    Let $G$ be a connected Spartan bipartite graph with bipartition $(A,B)$ where $A = \{a_1,\ldots,a_p\}$ and $B = \{b_1,\ldots,b_q\}$ (see also~\Cref{fig:spartanPM}). For the sake of contradiction, assume that $G$ does \emph{not} admit a perfect matching, i.e., without loss of generality let $q>k$. Let $S$ be a minimum vertex cover in $G$ of size $k$ such that the guards are placed on $S$ in the initial configuration. Let $M$ be any maximum matching. Without loss of generality, let the edges of $M$ be given by:
    $$\{e_1 = (a_1,b_1), \ldots, e_i = (a_i,b_i), e_{i+1} = (a_{i+1},b_{i+1}), \ldots, e_k = (a_k,b_k)\},$$
    
    where $S \cap A = \{a_1, \ldots, a_i\}$ and $S \cap B = \{b_{i+1}, \ldots, b_k\}$. 
    
    Since $M$ is not a perfect matching, $|V(G)|> 2k$ i.e. there exists a vertex in $G$ which is not an endpoint of any edge in $M$. Without loss of generality, let this vertex be $b_q$. Since $b_q$ is not an isolated vertex, it must have a neighbour on the $A$ side. Since $S$ is a vertex cover, this neighbour must lie in $S\cap A$. Without loss of generality, let this neighbour be $a_1$. If the attacker attacks $a_1b_q$, the guard on $a_1$ must come to $b_q$ and cannot move further. Now there are $k-1$ guards on the vertices of $M$ to protect $k$ matching edges. So no matter how these guards arrange themselves, at least one edge of $M$ must be vulnerable before the next attack. The attacker attacks this edge in the next move and wins thus contradicting the fact that $G$ is Spartan. Therefore, the graph $G$ must admit a perfect matching.     
\end{proof}

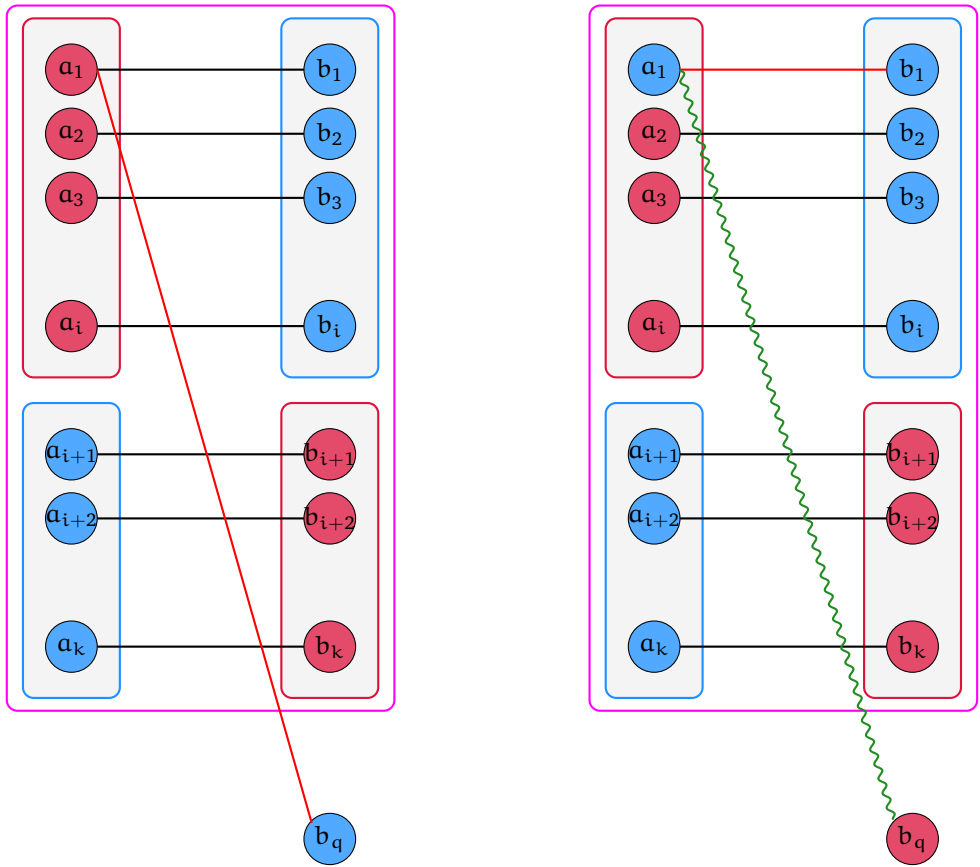
\begin{figure}[h]
    \begin{subfigure}[t]{0.45\textwidth}
    
        \centering
        
        \begin{tikzpicture}[scale=0.85]
        {
        
        \draw [Fuchsia,thick,rounded corners] (-1,2) rectangle (5,13);
        \draw [Crimson,thick,rounded corners,fill=LightGray!25] (-0.75,7.2) rectangle (0.75,12.8);
        \draw [DodgerBlue,thick,rounded corners,fill=LightGray!25] (3.25,7.2) rectangle (4.75,12.8);
        \draw [DodgerBlue,thick,rounded corners,fill=LightGray!25] (-0.75,2.2) rectangle (0.75,6.8);
        \draw [Crimson,thick,rounded corners,fill=LightGray!25] (3.25,2.2) rectangle (4.75,6.8);
        \draw[fill=Crimson!77] (0, 12) circle (0.4cm);
        \node at (0,12){$a_1$};
        \draw[fill=Crimson!77] (0, 11) circle (0.4cm);
        \node at (0,11){$a_2$};
        \draw[fill=Crimson!77] (0, 10) circle (0.4cm);
        \node at (0,10){$a_3$};
        \draw[fill=Crimson!77] (0, 8) circle (0.4cm);
        \node at (0,8){$a_i$};
        \draw[fill=DodgerBlue!77] (0, 6) circle (0.4cm);
        \node at (0,6){$a_{i+1}$};
        \draw[fill=DodgerBlue!77] (0, 5) circle (0.4cm);
        \node at (0,5){$a_{i+2}$};
        \draw[fill=DodgerBlue!77] (0, 3) circle (0.4cm);
        \node at (0,3){$a_k$};
        \draw[fill=DodgerBlue!77] (4, 12) circle (0.4cm);
        \node at (4,12){$b_1$};
        \draw[fill=DodgerBlue!77] (4, 11) circle (0.4cm);
        \node at (4,11){$b_2$};
        \draw[fill=DodgerBlue!77] (4, 10) circle (0.4cm);
        \node at (4,10){$b_3$};
        \draw[fill=DodgerBlue!77] (4, 8) circle (0.4cm);
        \node at (4,8){$b_i$};
        \draw[fill=Crimson!77] (4, 6) circle (0.4cm);
        \node at (4,6){$b_{i+1}$};
        \draw[fill=Crimson!77] (4, 5) circle (0.4cm);
        \node at (4,5){$b_{i+2}$};
        \draw[fill=Crimson!77] (4, 3) circle (0.4cm);
        \node at (4,3){$b_k$};
         \draw[fill=DodgerBlue!77] (4, 0) circle (0.4cm);
        \node at (4,0){$b_q$};
     \foreach \x in {3,5,6,8,10,11,12}
     \draw[thick] (0.4,\x) -- (3.6,\x);
     
     \draw[thick,Red] (0.4,12) -- (3.72,0.25);

        }
        \end{tikzpicture}
        
        \caption{A graph which does not have a perfect matching and the red edge is attacked}
        \label{fig:my_label-g}
    \end{subfigure}
    ~~~~~~~~~~~~
    \begin{subfigure}[t]{0.45\textwidth}
    
        \centering
        
        \begin{tikzpicture}[scale=0.85]
        {
         \tikzset{decoration={snake,amplitude=.4mm,segment length=2mm, post length=0mm,pre length=0mm}}
        
         \draw [Fuchsia,thick,rounded corners] (-1,2) rectangle (5,13);
        \draw [Crimson,thick,rounded corners,fill=LightGray!25] (-0.75,7.2) rectangle (0.75,12.8);
        \draw [DodgerBlue,thick,rounded corners,fill=LightGray!25] (3.25,7.2) rectangle (4.75,12.8);
        \draw [DodgerBlue,thick,rounded corners,fill=LightGray!25] (-0.75,2.2) rectangle (0.75,6.8);
        \draw [Crimson,thick,rounded corners,fill=LightGray!25] (3.25,2.2) rectangle (4.75,6.8);
        
        \draw[fill=DodgerBlue!77] (0, 12) circle (0.4cm);
        \node at (0,12){$a_1$};
        \draw[fill=Crimson!77] (0, 11) circle (0.4cm);
        \node at (0,11){$a_2$};
        \draw[fill=Crimson!77] (0, 10) circle (0.4cm);
        \node at (0,10){$a_3$};
        \draw[fill=Crimson!77] (0, 8) circle (0.4cm);
        \node at (0,8){$a_i$};
        \draw[fill=DodgerBlue!77] (0, 6) circle (0.4cm);
        \node at (0,6){$a_{i+1}$};
        \draw[fill=DodgerBlue!77] (0, 5) circle (0.4cm);
        \node at (0,5){$a_{i+2}$};
        \draw[fill=DodgerBlue!77] (0, 3) circle (0.4cm);
        \node at (0,3){$a_k$};
        \draw[fill=DodgerBlue!77] (4, 12) circle (0.4cm);
        \node at (4,12){$b_1$};
        \draw[fill=DodgerBlue!77] (4, 11) circle (0.4cm);
        \node at (4,11){$b_2$};
        \draw[fill=DodgerBlue!77] (4, 10) circle (0.4cm);
        \node at (4,10){$b_3$};
        \draw[fill=DodgerBlue!77] (4, 8) circle (0.4cm);
        \node at (4,8){$b_i$};
        \draw[fill=Crimson!77] (4, 6) circle (0.4cm);
        \node at (4,6){$b_{i+1}$};
        \draw[fill=Crimson!77] (4, 5) circle (0.4cm);
        \node at (4,5){$b_{i+2}$};
        \draw[fill=Crimson!77] (4, 3) circle (0.4cm);
        \node at (4,3){$b_k$};
         \draw[fill=Crimson!77] (4, 0) circle (0.4cm);
        \node at (4,0){$b_q$};
     \foreach \x in {3,5,6,8,10,11}
     \draw[thick] (0.4,\x) -- (3.6,\x);
     \draw[thick,Red] (0.4,12) -- (3.6,12);
     \draw[thick,decorate,ForestGreen] (0.4,12) -- (3.7,0.3);

        }
        \end{tikzpicture}

        \caption{After defending the attack, there are only $k-1$ guards to protect $k$ matching edges. So an edge will remain vulnerable.}
        \label{fig:my_label-f}
    \end{subfigure}
    \caption{Demonstrating that Spartan connected bipartite graphs must have perfect matching towards a proof of~\Cref{prop:spartanimpliespm}.}
    \label{fig:spartanPM}
    \end{figure}

So we now have that connected Spartan bipartite graphs $G$ have at least one perfect matching. From now on we denote the bipartitions of $G$ by $A=\{a_1,a_2,\ldots a_n\}$ and $B=\{b_1,b_2,\ldots b_n\}$ where each $a_i$ and $b_i$ have an edge between them so that we have a perfect matching $M=\bigcup_{i=1}^{n}\{a_i b_i\}$. Since $G$ is connected, these bipartitions are unique upto permutations. We now show that connected bipartite graphs that are Spartan do not have degree one vertices. This follows from known results (for example, as shown by~\cite{BCFPRW2019}), however, our argument is more direct and we make it explicit to make our presentation self-contained.

\begin{proposition}[No Degree One Vertex]\label{l1}
\label{prop:nodeg1}
Let $G$ be a connected bipartite graph with $|V(G)|>2$ and with a perfect matching. If $evc(G)=mvc(G)$, then $G$ has no degree one vertex.
\end{proposition}

\begin{proof}
Let $G$ be a connected bipartite graph with $|V(G)|>2$ and with a perfect matching. Assume that $evc(G)=mvc(G)$, and that $G$ has a degree $1$ vertex. We will show that this leads to a contradiction (c.f~\Cref{fig-nodegone}). Suppose $G$ has a degree $1$ vertex on the $A$ side, say $a_1$. Since the graph is connected, $b_1$ has an edge with some other vertex say $a_2$. If $b_1$ has a guard in the initial configuration, then attack $a_1b_1$. This will cause the guard on $b_1$ to move to $a_1$. (It is not possible that $a_1$ already had a guard because $a_1 b_1$ belongs to a perfect matching $M$ and we have only $n$ guards so some other edge in $M$ will be left with no guard). After this no other guard can come to $b_1$ for a similar reason. So without loss of generality we can assume that $a_1$ has a guard and $b_1$ has no guard in the initial configuration.

Since $b_1$ has no guard, $a_2$ must have a guard to protect $a_2 b_1$. Now attack the edge $a_2 b_1$. The guard on $a_2$ must move to $b_1$. The guard on $a_1$ cannot move anywhere. So we have an overprotected edge $a_1b_1$ which belongs to a perfect matching which means that some other edge in the matching must be vulnerable. So $n$ guards are not sufficient to protect the graph and hence $evc(G)\neq mvc(G)$.
\end{proof}

From \Cref{ele} and \Cref{l1}, we have the following.

\begin{corollary}\label{C1}
Let $G$ be a connected bipartite graph with more than two vertices. If $G$ is elementary, then $G$ cannot have a degree one vertex.
\end{corollary}

Next, we show that a bipartite graph with a cycle that visits every vertex is elementary.






\begin{proposition}[Adding edges to cycles]\label{cycle}
Consider a connected bipartite graph $G = (A \cup B, E)$ which is a cycle $a_1b_1a_2b_2\ldots a_kb_ka_1$ (where $k>1$). The graph $G'$ formed after adding any number of edges between $A$ and $B$ (preserving bipartiteness) will be elementary.
\end{proposition}

\begin{proof}
 Let $G$ be a connected bipartite graph $G$ which is a cycle $a_1b_1a_2b_2\ldots a_kb_ka_1$, where $k>1$ (see also~\Cref{fig-cycleedges}). Here the two sides of the bipartition are $A=\{a_1,a_2,\ldots,a_k\}$ and $B=\{b_1,b_2,\ldots,b_k\}$. 
 
 The graph $G$ has two perfect matchings which are given by $M=\{a_1b_1,a_2b_2,\ldots,a_k b_k\}$ and $M'=\{b_1a_2,b_2a_3,\ldots,b_{k-1}a_k,b_ka_1\}$. It is easily seen that any edge of $G$ lies in one of these two perfect matchings. 
 
 Suppose we add some edges from $A$ to $B$. Notice that the vertex set of the new graph $G'$ is same as the vertex set of $G$. In order to prove that $G'$ is elementary, it is sufficient to show that each of the newly added edges belongs to a perfect matching of $G'$.
 
 Consider a newly added edge $a_ib_j$. This clearly preserves bipartiteness. Here, $i\neq j,j+1$ because the edges $a_ib_i$ and $a_ib_{i-1}$ were already in $G$. 
 
 First, let $i>j+1$. Consider the matching: 
 $$M_1=\{a_ib_j,b_ia_{i+1},b_{i+1}a_{i+2},\ldots,b_{k-1}a_k,b_ka_1, b_1a_2,\ldots,b_{j-1}a_{j}, a_{j+1}b_{j+1}, a_{j+2}b_{j+2},\ldots,a_{i-1}b_{i-1}\}.$$ 
 
 Clearly $M_1$ is a perfect matching containing $a_ib_j$. 
 
 Now let $i<j$. Let $N_1=\{a_ib_j, b_ia_{i+1},b_{i+1}a_{i+2},\ldots,b_{j-1}a_j\}; N_2=\{a_1b_1,\ldots,a_{i-1}b_{i-1}\};$ and $N_3=\{a_{j+1}b_{j+1},\ldots,a_kb_k\}.$ 
 
 Consider the matching $M_2$ given by $N_1\cup N_2 \cup N_3$. Clearly $M_2$ is a perfect matching containing $a_ib_j$.
 
 This concludes our argument showing that $G'$ is elementary. 
 \end{proof}

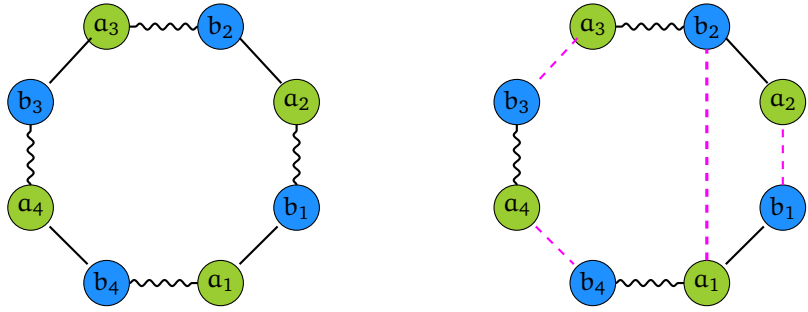
\begin{figure}[H]
\begin{subfigure}[t]{0.45\textwidth}
    \centering
    \begin{tikzpicture}
    {\tikzset{decoration={snake,amplitude=.4mm,segment length=2mm, post length=0mm,pre length=0mm}}
    
    \draw[fill=YellowGreen] (6, 3) circle (0.3cm);
    \node at (6,3){$a_1$};
     \draw[fill=DodgerBlue] (7, 4) circle (0.3cm);
    \node at (7,4){$b_1$};
     \draw[fill=YellowGreen] (7, 5.4) circle (0.3cm);
    \node at (7,5.4){$a_2$};
    \draw[fill=DodgerBlue] (6, 6.4) circle (0.3cm);
    \node at (6,6.4){$b_2$};
     \draw[fill=YellowGreen] (4.5, 6.4) circle (0.3cm);
    \node at (4.5,6.4){$a_3$};
    \draw[fill=DodgerBlue] (3.5, 5.4) circle (0.3cm);
    \node at (3.5,5.4){$b_3$};
     \draw[fill=YellowGreen] (3.5, 4) circle (0.3cm);
    \node at (3.5,4){$a_4$};
     \draw[fill=DodgerBlue] (4.5, 3) circle (0.3cm);
    \node at (4.5,3){$b_4$};
    \draw[thick] (6.25, 3.25)--(6.75, 3.75);
    \draw[thick] (6.8, 5.65)--(6.25, 6.25);
    \draw[thick] (4.3, 6.25)--(3.75, 5.65);
    \draw[thick] (3.75, 3.75)--(4.25, 3.25);
    \draw[thick,decorate] (4.8,3)--(5.7,3);
    \draw[thick,decorate] (7,4.3)--(7,5.1);
    \draw[thick,decorate] (5.7,6.4)--(4.8,6.4);
    \draw[thick,decorate] (3.5,4.3)--(3.5,5.1);
    }
    \end{tikzpicture}
    \caption{This figure shows a cycle with $k = 4$.}
    \label{fig:my_label-e}
\end{subfigure}
\begin{subfigure}[t]{0.45\textwidth}
    \centering
    \begin{tikzpicture}
    {\tikzset{decoration={snake,amplitude=.4mm,segment length=2mm, post length=0mm,pre length=0mm}}
    
    \draw[fill=YellowGreen] (6, 3) circle (0.3cm);
    \node at (6,3){$a_1$};
     \draw[fill=DodgerBlue] (7, 4) circle (0.3cm);
    \node at (7,4){$b_1$};
     \draw[fill=YellowGreen] (7, 5.4) circle (0.3cm);
    \node at (7,5.4){$a_2$};
    \draw[fill=DodgerBlue] (6, 6.4) circle (0.3cm);
    \node at (6,6.4){$b_2$};
     \draw[fill=YellowGreen] (4.5, 6.4) circle (0.3cm);
    \node at (4.5,6.4){$a_3$};
    \draw[fill=DodgerBlue] (3.5, 5.4) circle (0.3cm);
    \node at (3.5,5.4){$b_3$};
     \draw[fill=YellowGreen] (3.5, 4) circle (0.3cm);
    \node at (3.5,4){$a_4$};
     \draw[fill=DodgerBlue] (4.5, 3) circle (0.3cm);
    \node at (4.5,3){$b_4$};
    \draw[thick] (6.25, 3.25)--(6.75, 3.75);
    \draw[thick] (6.8, 5.65)--(6.25, 6.25);
    \draw[thick,dashed, Fuchsia] (4.3, 6.25)--(3.75, 5.65);
    \draw[thick,dashed,Fuchsia] (3.75, 3.75)--(4.25, 3.25);
    \draw[thick,decorate] (4.8,3)--(5.7,3);
    \draw[thick,dashed,Fuchsia] (7,4.3)--(7,5.1);
    \draw[thick,decorate] (5.7,6.4)--(4.8,6.4);
    \draw[thick,decorate] (3.5,4.3)--(3.5,5.1);
    \draw[very thick,dashed,Fuchsia] (6,3.3)--(6,6.1);
    }
    \end{tikzpicture}
    \caption{Each newly added edge is allowed: Suppose $a_1b_2$ is the new edge, it can be combined with the existing edges to form the perfect matching denoted by the dashed purple lines.}
    \label{fig:my_label-h}
\end{subfigure}
\caption{An example of a bipartite graph with a spanning cycle.}
\label{fig-cycleedges}
\end{figure}

We now introduce the terminology ``special subset'' to indicate that we are working with a subset of endpoints of some edges in a perfect matching.

\begin{definition}[Special Subset]\label{SS}
Let $G$ be a connected bipartite graph with a perfect matching $M=\{a_1b_1, a_2b_2,\ldots,a_nb_n\}$. Consider: $S=\{a_{i_1}, b_{i_1},a_{i_2},b_{i_2},\ldots,a_{i_k},b_{i_k}\}$ for some distinct $i_1,i_2,\ldots,i_k\in[n]$ and $2\leq k\leq n$. Then $S$ is said to be a \textbf{special subset} of $V(G)$.
\end{definition}

A ``special induced subgraph'' is a subgraph induced by a special subset.

\begin{definition}[Special Induced Subgraph]\label{SIG}
Let $G$ be a connected bipartite graph with a perfect matching and $S$ be a special subset of $V(G)$. The subgraph $G[S]$ induced by $S$ is called a \textbf{special induced subgraph} of $G$.
\end{definition}

We now note that edges allowed in special induced subgraphs of a graph $G$ are also allowed in $G$.

\begin{proposition}[Allowed Edge]\label{AE}
Let $G$ be a connected bipartite graph with a perfect matching $M$. Let $S$ be a special subset of $G$ and $H=G[S]$ be the corresponding special induced subgraph. An edge which is allowed in $H$ is also allowed in $G$.
\end{proposition}

\begin{proof}
 If $S=V(G)$, we are done. Otherwise without loss of generality, let $S$ be the set $\{a_1,b_1,a_2,b_2,\ldots,a_k,b_k\}$. Consider an edge $a_ib_j$ which is allowed in $H$ for some $i,j\in [k]$. Consider the matching $M_H$ which is a perfect matching of $H$ containing $a_ib_j$. Consider $M'=M_H\cup \{a_{i+1}b_{i+1},\ldots,a_kb_k\}$. Clearly $M'$ is a perfect matching of $G$ containing $a_ib_j$. Thus, $a_ib_j$ is allowed in $G$.
 \end{proof}

We now turn to a key definition: one of ``contracting'' a special elementary subset to a single edge. We refer the reader to~\Cref{fig:contraction} for an example of an application of this operation.

\begin{definition}[Matching Contraction Graph]\label{MCG}
Consider a connected bipartite graph $G$ with bipartition $A$ and $B$ and a perfect matching $M$. Let $S$ and $H=G[S]$ be a special subset of $V(G)$ and the corresponding special induced subgraph of $G$. If $H$ is elementary, replace $A\cap S$ by a new vertex $\alpha$ and $B\cap S$ by a new vertex $\beta$ in $G$. Add an edge $\alpha\beta$ in the new graph. For an edge $ab$ where $a\in A\backslash S$ and $b \in S$, add an edge $a\beta$ to the new graph. Similarly for an edge $ab$ where $a\in S$ and $b \in B\backslash S$, add an edge $\alpha b$ to the new graph. Keep all the other edges and vertices of $G$ unchanged. If there are any parallel edges, replace them by a single edge. The new graph $G_S$ thus obtained is called the \textbf{matching contraction graph} of $G$ w.r.t. $S$ and the procedure of obtaining a matching contraction graph is called as \textbf{contraction}.
\end{definition}

Note that if $G$ is a connected bipartite graph and $S$ is a special subset, then the matching contraction graph $G_S$ is also connected and bipartite.

\begin{figure}

\begin{subfigure}[t]{0.45\textwidth}

    \centering
    
    \begin{tikzpicture}
    {
      \draw [Fuchsia,thick,rounded corners] (-1.2,6.8) rectangle (5.2,11.2);   
    \draw [DodgerBlue,thick,rounded corners,fill=LightGray!25] (-0.75,7.2) rectangle (0.75,10.8);
    \draw [DodgerBlue,thick,rounded corners,fill=LightGray!25] (3.25,7.2) rectangle (4.75,10.8);
    \draw[fill=DodgerBlue!77] (0, 10) circle (0.4cm);
    \node at (0,10){$a_1$};
    \draw[fill=DodgerBlue!77] (0, 9) circle (0.4cm);
    \node at (0,9){$a_2$};
    \draw[fill=DodgerBlue!77] (0, 8) circle (0.4cm);
    \node at (0,8){$a_3$};
    \draw[fill=DodgerBlue!77] (0, 6) circle (0.4cm);
    \node at (0,6){$a_{4}$};
    \draw[fill=DodgerBlue!77] (0, 5) circle (0.4cm);
    \node at (0,5){$a_{5}$};
    \draw[fill=DodgerBlue!77] (0, 4) circle (0.4cm);
    \node at (0,4){$a_6$};
     \draw[fill=DodgerBlue!77] (0, 3) circle (0.4cm);
    \node at (0,3){$a_7$};
    \draw[fill=DodgerBlue!77] (4, 10) circle (0.4cm);
    \node at (4,10){$b_1$};
    \draw[fill=DodgerBlue!77] (4, 9) circle (0.4cm);
    \node at (4,9){$b_2$};
    \draw[fill=DodgerBlue!77] (4,8) circle (0.4cm);
    \node at (4,8){$b_3$};
    \draw[fill=DodgerBlue!77] (4, 6) circle (0.4cm);
    \node at (4,6){$b_{4}$};
    \draw[fill=DodgerBlue!77] (4, 5) circle (0.4cm);
    \node at (4,5){$b_{5}$};
    \draw[fill=DodgerBlue!77] (4, 4) circle (0.4cm);
    \node at (4,4){$b_6$};
     \draw[fill=DodgerBlue!77] (4, 3) circle (0.4cm);
    \node at (4,3){$b_7$};
 \foreach \x in {3,4,5,6,8,9,10}
 \draw[thick] (0.4,\x) -- (3.6,\x);
 \draw[thick] (0.4,10) -- (3.6,8);
  \draw[thick] (0.4,8) -- (3.6,10);
 \draw[thick] (0.4,9) -- (3.6,10);
 \draw[thick] (0.4,10) -- (3.6,9);
 \draw[thick] (0.4,10) -- (3.6,6);
 \draw[thick] (0.4,8) -- (3.6,4);
 \draw[thick] (0.4,9) -- (3.6,3);
 \draw[thick] (0.4,6) -- (3.6,10);
 \draw[thick] (0.4,5) -- (3.6,10);
 \draw[thick] (0.4,5) -- (3.6,9);
 \draw[thick] (0.4,4) -- (3.6,8);
 
    }
    \end{tikzpicture}

    \caption{A graph $G$ with special elementary subset $S=\{a_1,b_1,a_2,b_2,a_3,b_3\}$.}
    \label{fig:my_label-b}
\end{subfigure}
~~~~~~~~~~~~~~~
\begin{subfigure}[t]{0.45\textwidth}

    \centering
    
    \begin{tikzpicture}
    {
    \draw[fill=Cyan] (0, 8) circle (0.6cm);
    \node at (0,8){$a_3$};
    \draw[fill=DodgerBlue!77] (0, 6) circle (0.4cm);
    \node at (0,6){$a_{4}$};
    \draw[fill=DodgerBlue!77] (0, 5) circle (0.4cm);
    \node at (0,5){$a_{5}$};
    \draw[fill=DodgerBlue!77] (0, 4) circle (0.4cm);
    \node at (0,4){$a_6$};
     \draw[fill=DodgerBlue!77] (0, 3) circle (0.4cm);
    \node at (0,3){$a_7$};
    \draw[fill=Cyan] (4,8) circle (0.6cm);
    \node at (4,8){$b_3$};
    \draw[fill=DodgerBlue!77] (4, 6) circle (0.4cm);
    \node at (4,6){$b_{4}$};
    \draw[fill=DodgerBlue!77] (4, 5) circle (0.4cm);
    \node at (4,5){$b_{5}$};
    \draw[fill=DodgerBlue!77] (4, 4) circle (0.4cm);
    \node at (4,4){$b_6$};
     \draw[fill=DodgerBlue!77] (4, 3) circle (0.4cm);
    \node at (4,3){$b_7$};
 \foreach \x in {3,4,5,6}
 \draw[thick] (0.4,\x) -- (3.6,\x);
 \draw[thick] (0.6,8) -- (3.4,8);
 \draw[thick] (0.6,8) -- (3.6,6);
 \draw[thick] (0.6,8) -- (3.6,4);
 \draw[thick] (0.6,8) -- (3.6,3);
 \draw[thick] (0.4,6) -- (3.4,8);
 \draw[thick] (0.4,5) -- (3.4,8);
 \draw[thick] (0.4,4) -- (3.4,8);
 
    }
    \end{tikzpicture}

    \caption{The matching contraction graph of $G$ w.r.t. $S$.}
    \label{fig:my_label-c}
\end{subfigure}
    \caption{Demonstrating the operation of contracting a special elementary subset to a single edge. For a more elaborate example involving a sequence of contractions, see Figure 12.}
    \label{fig:contraction}
\end{figure}
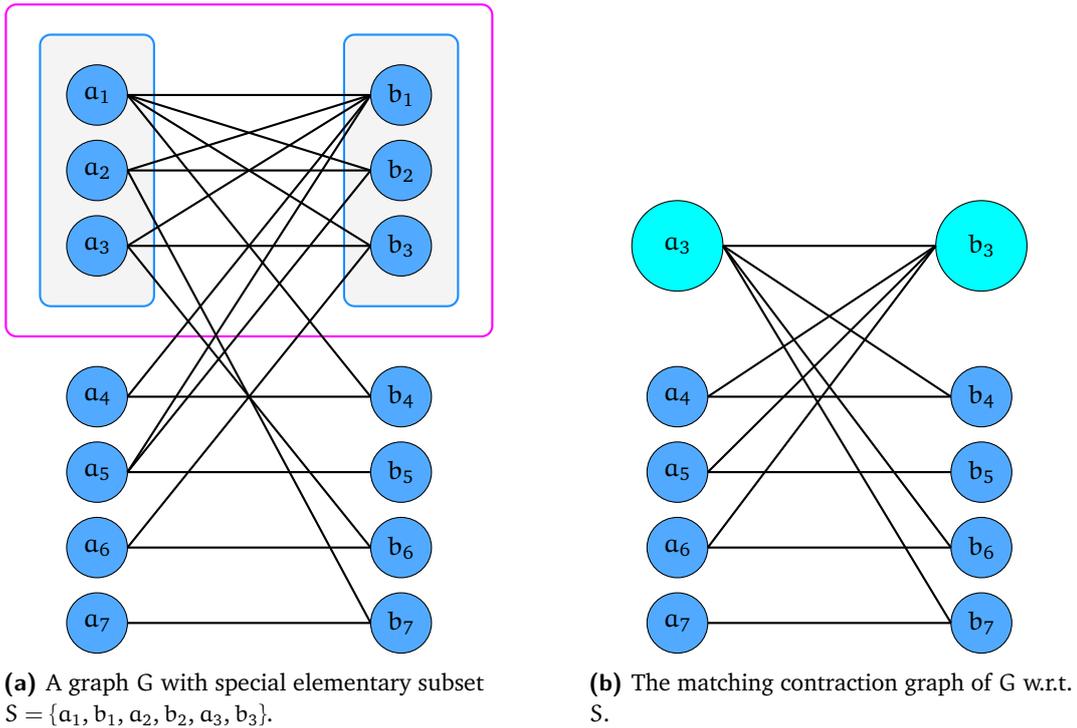

It can be seen that if the special subset $S$ has size $2k$, then the matching contraction graph has $2(n-k+1)$ vertices, a perfect matching of size $n-k+1$ and thus a minimum vertex cover of size $n-k+1$.

\begin{definition}[Special elementary subset]
Let $G$ be a connected bipartite graph with a perfect matching. A special subset $S$ such that $G[S]$ is elementary is called a \textbf{special elementary subset}.
\end{definition}

We consider a series of contractions over a graph $G$. Each time we contract a special elementary subset $S$ of $G$. We rename each vertex in the new graph as the largest indexed vertex in the original graph. We associate a label with each vertex, which carries the names of the vertices of the original graph that were contracted to this new vertex. We keep on repeating this procedure until we reach a graph which cannot be contracted further. We denote this graph by $G_{\bar{S}}$. We fix a perfect matching $M$ in $G$ and in each contracted graph, we get a matching given by $M\setminus E(G[S])\cup \{\alpha \beta\}$ where $\alpha$ and $\beta$ denote the newly obtained vertices on the $A$ and $B$ side respectively. Whenever we say ``matched partner'' or ``matching edge'' in the new graph, unless mentioned otherwise, we will be referring to this matching.

\begin{lemma}[Matching Contraction Graph Property]\label{MCP}

Suppose that $G$ is a connected bipartite graph with a perfect matching $M$. Then, if $G$ is Spartan, then $G_{\bar{S}}$ is Spartan.
\end{lemma}

\begin{proof}
If $G_{\bar{S}}$ is a single edge, then $G$ was elementary to begin with, as each expansion preserves the ``elementary-ness'' by \Cref{MC}. And since a single edge is Spartan, $G_{\bar{S}}$ is Spartan in this case.

Consider the case where $G_{\bar{S}}$ is not a single edge. Now by \Cref{stru}, $G_{\bar{S}}$ must have a degree-$1$ vertex. Without loss of generality, we assume that this degree-$1$ vertex is on the $B$ side and let us denote this vertex by $\beta_1$ and its matched partner by $\alpha_1$. Also, note that $G_{\bar{S}}$ is not Spartan by \Cref{prop:nodeg1}. We show that $G$ was also not Spartan in this case.

Let $B_1$ be the set of vertices in $G$ in the label of $\beta_1$, i.e, which were contracted to get $\beta_1$ and similarly $A_1$ be the set of vertices in $A$ which were contracted to get $\alpha_1$ and let $|A_1|=|B_1|=k$. Note that this contraction may not have happened in one step, but over a series of steps. Notice that $G\setminus \{A_1 \cup B_1\}$ is non-empty, as $G_{\bar{S}}$ is not a single edge. Note that no vertex in $B_1$ can be adjacent to a vertex in $A\setminus A_1$. Suppose some vertex $b\in B_1$  is adjacent to a vertex $a\in A\setminus A_1$, this will result in an edge $a\beta_1$ in $G_{\bar{S}}$ which contradicts the fact that $\beta_1$
is a degree-$1$ vertex. Also, we know that after a series of contractions, $G[A_1\cup B_1]$ got contracted to a single edge. By \Cref{MC}, $G[A_1\cup B_1]$ is elementary. Therefore, any minimum sized vertex cover of $G$ will contain all the $k$ vertices of $A_1$ and no vertex from $B_1$ or all the $k$ vertices of $B_1$ and no vertex from $A_1$. 

Now consider any initial configuration of guards on $G$ with $n$ guards. As seen above, either all the vertices of $B_1$ have guards or all the vertices of $A_1$ have guards and only one of these can happen. If all the vertices of $A_1$ are occupied by guards, the attacker attacks a matched edge $a_ib_i$ such that $a_i\in A_1$. The guard on $a_i$ is forced to move to $b_i$ and cannot move any further. If all the guards cannot reconfigure to form a vertex cover, then $G$ is not Spartan and hence we are done. Otherwise, all the guards must reconfigure in such a way that all the vertices in $B_1$ have guards and all the vertices in $A_1$ do not have guards.
Thus without loss of generality, we can assume that all the vertices in $B_1$ have guards and all the vertices in $A_1$ do not have guards.

Now recall that no vertex in $B_1$ has a neighbour in $A\setminus A_1$. Since $G$ is connected, there exists $a_p\in A_1$ and $b_q \in B\setminus B_1$ such that $a_pb_q\in E(G)$. Since $a_p$ does not have a guard, $b_q$ must have a guard as the guards are occupying a vertex cover. Now suppose the attacker attacks the edge $a_pb_q$, the guard on $b_q$ must move to $a_p$. Since all the vertices in $B_1$ have guards but none of them has a neighbour in $A\setminus A_1$ and no vertex in $A_1$ had a guard before this guard just moved to $a_p$, now no guard can move from $\{A_1 \cup B_1\}$ to $V(G)\setminus A_1$. Thus there will be $k+1$ guards on $k$ matching edges in $G[A_1\cup B_1]$. Therefore, some edge in $G[V(G)\setminus \{A_1 \cup B_1\}]$ will be vulnerable and can be attacked in the next move. Hence, $G$ is not Spartan.
\end{proof}

\emph{Remark:} We have actually shown that any bipartite graph $G$ which is not ``essentially elementary'' can be destroyed by the attacker in at most three moves when the defender has $mvc(G)$ many guards.

\begin{sidewaysfigure}
    \includegraphics[scale=0.35]{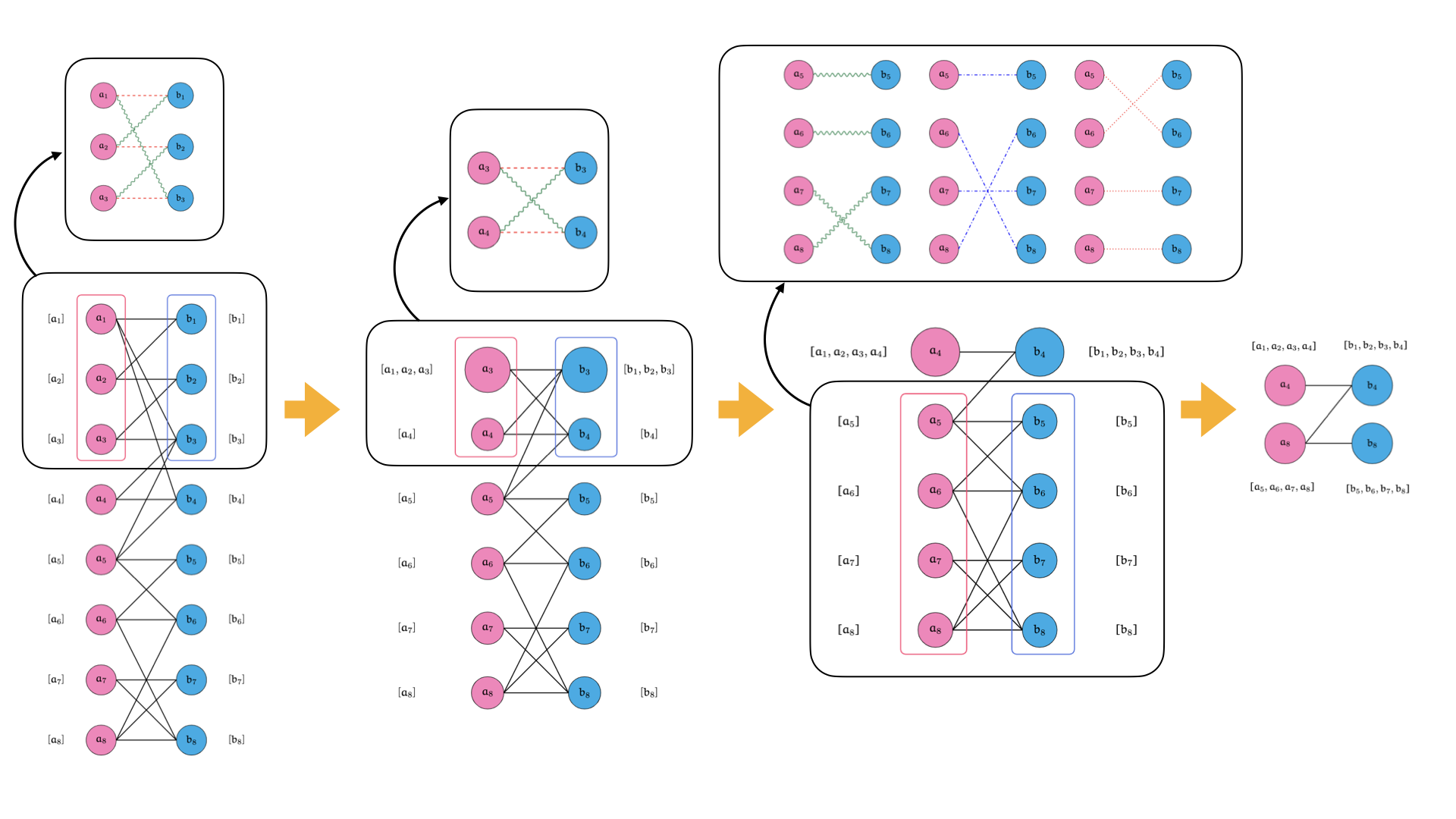}
    \caption{Demonstrating a sequence of contractions.}
    \label{contraction1}
\end{sidewaysfigure}

\begin{sidewaysfigure}
    \includegraphics[scale=0.35]{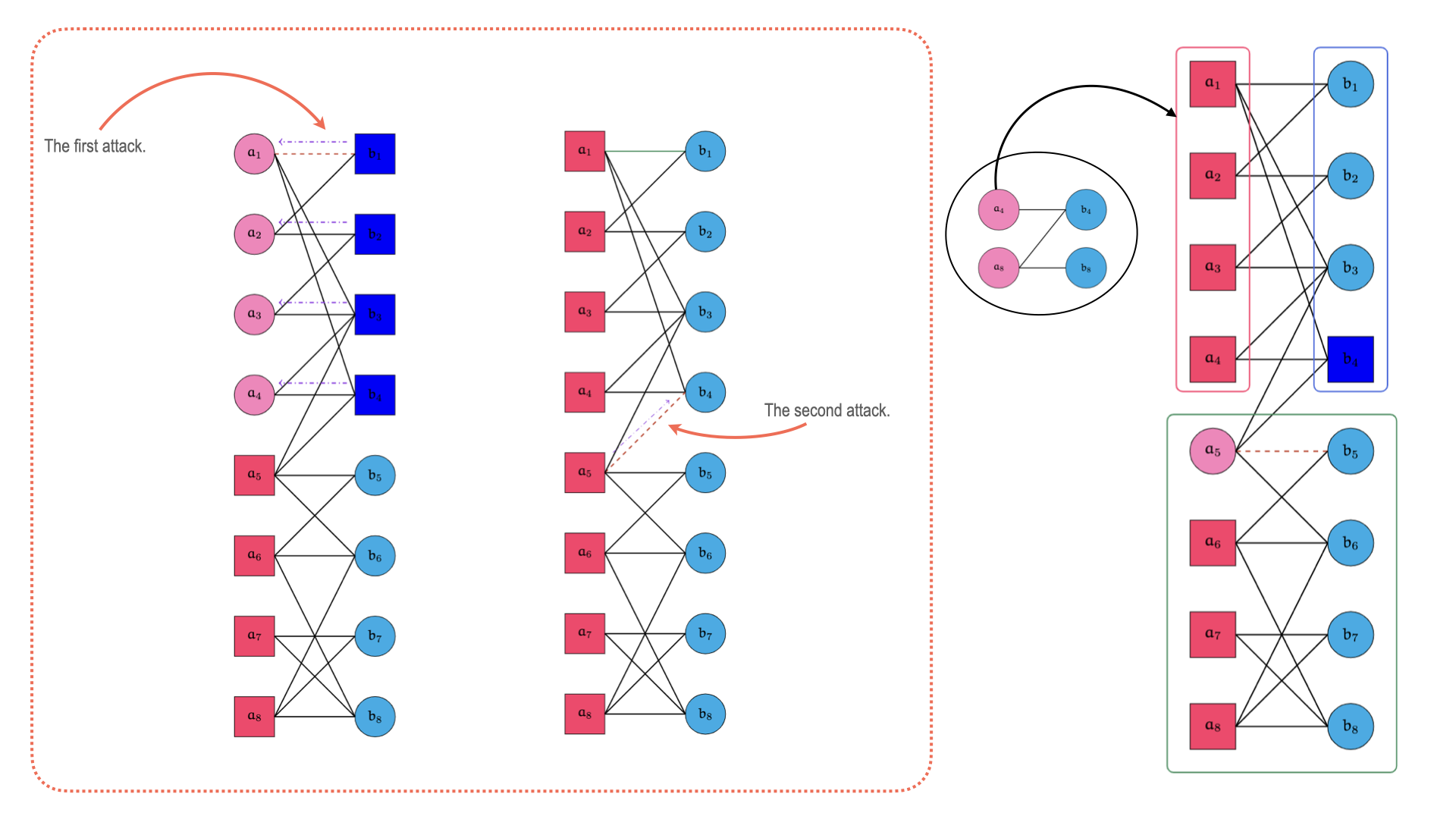}
    \caption{Demonstrating how attacks are mimicked.}
    \label{contraction2}
\end{sidewaysfigure}

\begin{definition}[Maximal Contraction Graph]\label{MaxG}
Let $G$ be a connected bipartite graph with a perfect matching and there does not exist any special elementary subset $S$ of $V(G)$. Then $G$ is said to be a \textbf{maximal contraction graph.}
\end{definition}

Note that by definition the size of a special subset is always more than $2$ and hence if $|V(G)|=2$, i.e., $G$ is a single edge, then $G$ is a maximal contraction graph as $G$ has no special subset and hence no special elementary subset. Note that the sequence of contractions above ends in a maximal contraction graph that is in fact an edge.
 
\begin{lemma}[Maximal Contraction Graph]\label{MC}
 Let $G$ be a connected bipartite graph with a perfect matching $M$. Let $S$ be a special elementary subset of $V(G)$ and $G_S$ be the corresponding matching contraction graph. Then, if $G_S$ is elementary, then $G$ is also elementary.
\end{lemma}
\begin{proof}

Let $G = (A \cup B,E)$. To show that $G$ is elementary, we show that $A$ and $B$ are the only optimal vertex covers of $G$. We assume that the edges of the perfect matching in $G$ are $\{a_{1}b_{1}, a_{2}b_{2}, \cdots, a_{n}b_{n}\}$ and that the special subset $S$ is induced by the vertices based on the edges $\{a_1b_1, \cdots, a_kb_k\}$. Let $\alpha$ and $\beta$ denote the endpoints of the edge in $G_S$ created by the contraction of $S$. Finally, let $(A^\prime,B^\prime)$ denote the partition of $G_S$.

Note that $G[S]$ is assumed to be elementary. If $X$ is an optimal vertex cover for $G$, then $X \cap S$ is a vertex cover for $G[S]$ of size $k$, since $G \setminus S$ has a matching of size $(n-k)$. Observe that:

\begin{itemize}
    \item either $X \cap A = \{a_1,\cdots,a_k\}$ and $X \cap B = \emptyset$, or
    \item $X \cap A = \emptyset$ and $X \cap B = \{b_1,\cdots,b_k\}$,
\end{itemize}

because any other subset of $k$ vertices that forms a valid vertex cover for $G[S]$ would contradict the assumption that $G[S]$ is elementary.

Now assume that $G$ has a vertex cover $X$ such that $X \neq A$ and $X \neq B$. Without loss of generality, assume that $X \cap A = \{a_1,\cdots,a_k\}$ (the argument for the other case is symmetric). We let $X_A := X \cap \{a_{k+1},\cdots,a_n\}$ and $X_B = X \cap \{b_{k+1},\cdots,b_n\}$. Note that $X_B \neq \emptyset$ because $X \neq A$. 

We now claim that $X^\prime := X_A \cup X_B \cup \{\alpha\}$ is an optimal vertex cover in $G_S$ that is different from both $A^\prime$ and $B^\prime$. It is clear that $|X^\prime| = (n-k+1)$, thus the size of $X^\prime$ is $mvc(G_S)$. Further, since $\alpha \in A^\prime$ and $X_B \neq \emptyset$, $X^\prime \neq A^\prime$ and $X^\prime \neq B^\prime$. It remains to be shown that $X^\prime$ covers all edges in $G_S$.

Note that if any edge $a_i b_j$ in $G_S[\{a_{k+1},\cdots,a_n\} \cup  \{b_{k+1},\cdots,b_n\}]$ is not covered by $X^\prime$, then the $a_i b_j$ is also not covered by $X$ in $G$. All edges incident on $\alpha$ are also covered. Now suppose an edge of the form $\beta a_\ell$ is not covered by $X^\prime$ for some $\ell \in \{a_{k+1},\cdots,a_n\}$. Note that for this edge to be present in $G_S$, by the definition of the contraction operation, there must have been an edge of the form $b_i a_\ell$, for some $i \in [k]$. Note that if $\beta a_\ell$ is not covered by $X^\prime$ then $b_i a_\ell$ is not covered by $X$ in $G$ either, which contradicts our assumption that $X$ was a vertex cover in $G$. 

This shows that $G_S$ has an optimal vertex cover different from both $A^\prime$ and $B^\prime$, but this contradicts our assumption that $G_S$ was elementary to begin with. Therefore it must be the case that $G$ is also elementary.
  \end{proof}

\begin{definition}[Alternating cycle]\label{AC}
Let $G$ be a connected bipartite graph with no degree $1$ vertex and a perfect matching $M=\{a_1b_1,a_2b_2,\ldots a_nb_n\}$. We define a cycle $a_{i_1}b_{i_1}a_{i_2}b_{i_2}\ldots b_{i_k}a_{i_k}a_{i_1}$ in $G$ as an \textbf{alternating cycle} where $i_1,i_2,\ldots,i_k\in [n]$ and $2\leq k\leq n$.
\end{definition}

\begin{proposition}[No alternating cycle]\label{NAC}
Let $G$ be a maximal contraction graph. Then $G$ cannot contain any alternating cycle.
\end{proposition}

\begin{proof}
 Let $G$ be a maximal contraction graph. It is clear that $G$ is a connected bipartite graph with a perfect matching. Let $M=\{a_1b_1,a_2b_2,\ldots,a_nb_n\}$ be a perfect matching of $G$. If $n=1$, $G$ cannot contain any cycle and hence cannot contain any alternating cycle. 
 Let $n\geq 2$ and let $C=a_{i_1}b_{i_1}a_{i_2}b_{i_2}\ldots b_{i_k}a_{i_k}a_{i_1}$ be an alternating cycle in $G$. Let $S=V(C)$. Clearly $S$ is a special subset of $V(G)$ and by \Cref{cycle} and \Cref{AE}, $S$ is a special elementary subset of $V(G)$. This contradicts the fact that $G$ is a maximal contraction graph. Thus $G$ cannot have an alternating cycle. 
\end{proof}

\begin{proposition}[Structure of the Maximal Contraction]\label{stru}
Let $G$ be a maximal contraction graph. Then $G$ must have a degree $1$ vertex.
\end{proposition}
\begin{proof}
 Let $G$ be a maximal contraction graph. It is clear that $G$ is a connected bipartite graph with a perfect matching. Let $M=\{a_1b_1,a_2b_2,\ldots,a_nb_n\}$ be a perfect matching of $G$. If $n=1$, then $G$ is a single edge and thus $G$ has a degree $1$ vertex. Now let $n\geq 2$ and assume that $G$ has no degree $1$ vertices. Consider the set $S=\{a_1,b_1\}$. Since $b_1$ is not a degree $1$ vertex, we are done. So there exits $a_{i_1}\neq a_1$ which is a neighbor of $b_1$. Add $a_{i_1}$ and $b_{i_1}$ to $S$. Now for each $b_{i_p}$ for $p=1,2,\ldots$, there exists $a_{p+1}$ which is a neighbor of $b_{i_p}$ because $G$ has no degree $1$ vertex or alternating cycle. Add $a_{p+1}$ and $b_{p+1}$ to $S$. Thus $S$ will be an infinite set. But $S\subset V(G)$ which is finite. Hence we have a contradiction.  Thus $G$ must have a degree $1$ vertex.
 \end{proof}

To sum up the proof of~\Cref{conv}: consider any connected bipartite graph $G$. First we have shown that if $G$ is Spartan, it must have a perfect matching and no degree one vertices. Therefore, a maximal contraction graph derived from $G$ must be an edge. But note that an edge is elementary, and running the contractions in reverse to recover $G$ preserves the property of the graph being elementary, and we have the desired conclusion. We remark that our proof shows that every connected bipartite graph with a Hamiltnonian cycle is Spartan. However, it turns out that the converse is not true in the sense that there exist connected bipartite graphs that are Spartan but do not have Hamiltonian cycles: for instance, Figure~\ref{fig:my_label2} presents an example of a connected bipartite Spartan graph that is not Hamiltonian.

\begin{figure}[H]
    \centering
    \begin{tikzpicture}
    {\tikzset{decoration={snake,amplitude=.4mm,segment length=2mm, post length=0mm,pre length=0mm}}
    
    \draw[fill=YellowGreen] (6, 3) circle (0.3cm);
    \node at (6,3){$b_1$};
     \draw[fill=DodgerBlue] (7, 4) circle (0.3cm);
    \node at (7,4){$a_2$};
     \draw[fill=YellowGreen] (7, 5.4) circle (0.3cm);
    \node at (7,5.4){$b_2$};
    \draw[fill=DodgerBlue] (6, 6.4) circle (0.3cm);
    \node at (6,6.4){$a_3$};
     \draw[fill=YellowGreen] (4, 6.4) circle (0.3cm);
    \node at (4,6.4){$b_3$};
    \draw[fill=DodgerBlue] (4, 3) circle (0.3cm);
    \node at (4,3){$a_1$};
   \draw[fill=YellowGreen] (5,4) circle (0.3cm);
  \node at (5,4){$b_4$};
    \draw[fill=DodgerBlue] (5,5.4) circle (0.3cm);
    \node at (5,5.4){$a_4$};
    \draw[thick] (6.25, 3.25)--(6.75, 3.75);
    \draw[thick] (6.8, 5.65)--(6.25, 6.25);
    \draw[thick] (4.3,3)--(5.7,3);
    \draw[thick] (7,4.3)--(7,5.1);
    \draw[thick] (5.7,6.4)--(4.3,6.4);
    \draw[thick] (4,3.3)--(4,6.1);
    \draw[thick]  (4.2,6.2)--(4.9,5.7);
    \draw[thick] (5.3,4)--(6.7,4);
     \draw[thick] (5,5.1)--(5,4.3);
    }
    \end{tikzpicture}
    \caption{An example of an elementary bipartite graph where the two independent sets are $A=\{a_1,a_2,a_3,a_4\}$ and $B=\{b_1,b_2,b_3,b_4\}$ and all the edges belong to one of the three perfect matchings: $M_1=\{a_1b_1,a_2b_2,a_3b_3,a_4b_4\},M_2=\{a_1b_3,a_3b_2,a_2b_1,a_4b_4\}$ or $M_3=\{a_1b_1,a_2b_4,a_4b_3,a_3b_2\}$. It can be verified that this graph does not have a Hamiltonian cycle.}
    \label{fig:my_label2}
\end{figure}
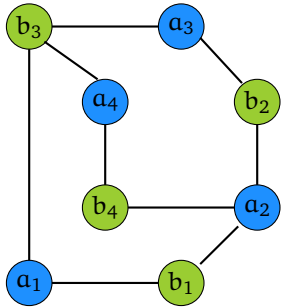

\section{EVC with Extra Steps}

Almost all of the proofs in the previous sections rely crucially on the fact that a guard can only move one step after each attack. This gives rise to the question that what happens if multiple moves are allowed? If retracing of steps is allowed, then any guard can clean up an attack and come back to their original position. Thus this problem will the same as the vertex cover problem and hence it is not so interesting. 

We define a variant of the \EVC{} problem that we call \emph{New Eternal Vertex Cover}. Again we have a two player game with one player as ``the attacker'' and the other player as ``the defender''. Just like the \EVC{} problem, the defender initially places the guards on some of the vertices of the graph (of his choice). The attacker attacks an edge. In response to the attack, the defender can move each guards for at most \emph{two} steps without retracing. The constraint is however that the defense of the attacked edge must happen in the first move itself, i.e., after the attack, at least one guard who was present on one of the endpoints of the graph must move across the attacked  edge. If such a movement is not possible after a finite sequence of attacks, the attacker wins; otherwise, if the defender has a strategy to defend the graph $G$ against an infinite sequence of attacks, the defender wins. There can only be one guard per vertex in the configuration before the attack and the configuration after the attack is defended. However, more than one guard can cross a vertex during the reconfiguration.

We define the \emph{New Eternal Vertex Cover Number} of a graph $G$ as the smallest number of guards required in this new setting such that the defender has a winning strategy. We denote this number by $nevc(G)$. It is clear that since the defense must happen in the first step of a guard, if the vertices occupied by guards do not form a vertex cover, the attacker wins. Therefore, we have $mvc(G)\leq nevc(G)$ for any graph $G$. Also, since $evc(G)$ many guards can defend an infinite sequence of attacks on a graph $G$ with each guard moving just one step after each attack, it can be seen that for any graph $G$, we have $nevc(G) \leq evc(G)$.

\emph{Remark:} These two inequalities imply that for any Spartan graph $G$, we have $mvc(G)=nevc(G)=evc(G)$. However, there are graph classes where these inequalities are not strict.

We show that computing the \emph{New Eternal Vertex Cover} number for a given graph $G$ is \NPH{} in the next lemma. We use a reduction from the \emph{Vertex Cover} problem. Given an input instance $(G,k)$ (where $|V(G)|=n$ and $k<n-1$) of the \emph{Vertex Cover} problem, we construct an equivalent instance of the \emph{New Eternal Vertex Cover} problem by adding a global vertex $\star$ to the graph $G$, i.e., make the new vertex $\star$ adjacent to every other vertex in $G$. It turns out that $G$ has a vertex cover of size $k$ if and only if the defender has a winning strategy in the new setting using $k+1$ guards.  

\begin{lemma}
The \emph{New Eternal Vertex Cover} problem is \NPH{}.
\end{lemma}

\begin{proof}   
    Given an input instance $(G,k)$ (where $|V(G)|=n$ and $k<n-1$) of the \emph{Vertex Cover} problem, we construct an equivalent instance of the \emph{New Eternal Vertex Cover} problem in polynomial time.
    
    Construct a graph $H$ by adding a global vertex $\star$ to the graph $G$, i.e., make the new vertex $\star$ adjacent to every other vertex in $G$. We show that $G$ has a vertex cover of size $k$ if and only if the defender has a winning strategy in the new setting using $k+1$ guards.

    \begin{figure}
        \centering
        \begin{tikzpicture}[scale=0.5]
            \draw [RoyalBlue,thick,rounded corners] (-5,0) rectangle (0,5); 
             \draw[fill=BlueViolet] (3,2.5) circle (0.5cm);
             \node at (3,2.5){$\star$};
            \draw[fill=Orchid] (-3,1) circle (0.3cm);
            
            \draw[fill=Orchid] (-4,2) circle (0.3cm);
            \draw[fill=Orchid] (-2,2) circle (0.3cm);
            \draw[fill=Orchid] (-3,4) circle (0.3cm);
            \draw[fill=Orchid] (-1,4) circle (0.3cm);
            \draw[fill=Orchid] (-0.8,0.8) circle (0.3cm);
            \draw[thick](-3.9,2.3)--(-3.2,3.8);
            \draw[thick](-2.7,4)--(-1.3,4);
            \draw[thick](-2.9,3.7)--(-2.2,2.2);
            \draw[thick](-1.1,3.7)--(-1.8,2.2);
            \draw[thick](-3.7,2)--(-2.3,2);
            \draw[thick](-3,1.3)--(-3,3.7);
            \draw[thick](-0.8,1.1)--(-1,3.7);
            \draw[thick,dashed](2.5,2.5)--(0,4);
            \draw[thick,dashed](2.5,2.5)--(0,3);
            \draw[thick,dashed](2.5,2.5)--(0,2);
            \draw[thick,dashed](2.5,2.5)--(0,1);   
        \end{tikzpicture}    
        \caption{Construction of the reduced instance from Vertex Cover to the New Eternal Vertex Cover Problem. Graph $G$ is given in the blue box and the new graph $H$ is obtained by adding a global vertex $\star$ to $G$.}
        \label{fig:my_label}
    \end{figure}
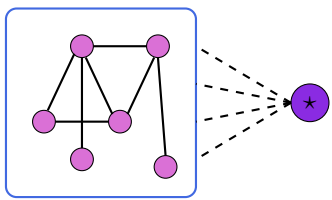
    
    Suppose $G$ has a vertex cover of size $k$, then consider the following arrangement of $k+1$ guards: Place one guard on each vertex of a vertex cover (of size $k$) and one guard on the global vertex $\star$. Now if the attacker attacks an edge which has guards on both of its endpoints, the guards exchange their positions and the configuration is restored.
    
    Suppose an edge $uv$ is attacked such that $u$ belongs to the vertex cover and $v$ does not, i.e., $u$ has a guard and $v$ does not, then the guard from $u$ goes to $v$ and the guard from $\star$ comes to $u$. The guard which is now on $u$ can move one more step and thus moves to $\star$. Thus the initial configuration is restored. This proves the forward direction.
    
    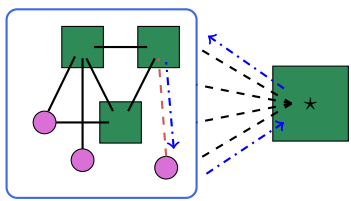
\begin{figure}[H]
        \centering
        \begin{tikzpicture}[square/.style={regular polygon,regular polygon sides=4},scale=0.5]
            \draw [RoyalBlue,thick,rounded corners] (-7,0) rectangle (-2,5); 
             
             \node at (1,2.5)[square,draw,fill=SeaGreen,minimum size=40pt](){$\star$};
            \draw[fill=Orchid] (-5,1) circle (0.3cm);
            \draw[fill=Orchid] (-6,2) circle (0.3cm);
            \draw[fill=Orchid] (-2.8,0.8) circle (0.3cm);
            \node at (-5,4)[square,draw,fill=SeaGreen,minimum size=22pt](){};
            \node at (-3,4)[square,draw,fill=SeaGreen,minimum size=22pt](){};
             \node at (-4,2)[square,draw,fill=SeaGreen,minimum size=22pt](){};
            
             \draw[thick](-5.9,2.3)--(-5.2,3.75);
            \draw[thick](-4.7,4)--(-3.3,4);
            \draw[thick](-4.9,3.7)--(-4.2,2.3);
            \draw[thick](-3.1,3.7)--(-3.8,2.3);
            \draw[thick](-5.7,2)--(-4.3,2);
            \draw[thick](-5,1.3)--(-5,3.7);
            \draw[thick,IndianRed,dashed](-2.8,1.1)--(-3,3.7);
            \draw[thick,dashed](0.5,2.5)--(-2,4);
            \draw[thick,dashed](0.5,2.5)--(-2,3);
            \draw[thick,dashed](0.5,2.5)--(-2,2);
            \draw[thick,dashed](0.5,2.5)--(-2,1);  
          \draw[thick,<-,Blue,dash dot](-2.6,1.3)--(-2.8,3.6);  
          \draw[thick,<-,Blue,dash dot](0.3,2)--(-1.8,0.6);
           \draw[thick,->,Blue,dash dot](0.3,2.9)--(-1.7,4.3);
        \end{tikzpicture}    
        \caption{Suppose any edge in $G$ with exactly one endpoint in the vertex cover is attacked, the guard from this vertex moves outside the vertex cover along the attacked edge. The guard on the global vertex comes to the vertex which was just vacated and the vertex cover is now restored. The guard who had come just outside the vertex cover will now move to the global vertex. Thus the configuration is restored.}
    \label{fig:my_label}
    \end{figure}
        
For the reverse direction, consider an initial configuration of guards on $H$ in a winning strategy with $k+1$ guards. The vertices occupied by guards must form a vertex cover of $H$ otherwise the attacker wins in the first move itself. Any vertex cover of $H$ with $k+1$ guards must contain $\star$ because $k+1<n$. Therefore, there must be $k$ guards on the vertices of $G$. If their positions do not form a vertex cover, the attacker can attack an edge with both the endpoints unoccupied and win. Thus $G$ has a vertex cover of size $k$.
\end{proof}

\begin{lemma}\label{sufficient}
For any connected graph $G$, the defender has a winning strategy using $mvc(G)+1$ guards, i.e., $nevc(G)\leq mvc(G)+1$ for any graph $G$.
\end{lemma}
\begin{proof}
For any graph $G$, we give a winning strategy for the defender in the new setting using $mvc(G)+1$ guards. We also refer the reader to~\Cref{fig:defense}.

\begin{figure}
    \centering
    \begin{subfigure}[b]{0.49\textwidth}
        \centering
        \includegraphics[width=\linewidth]{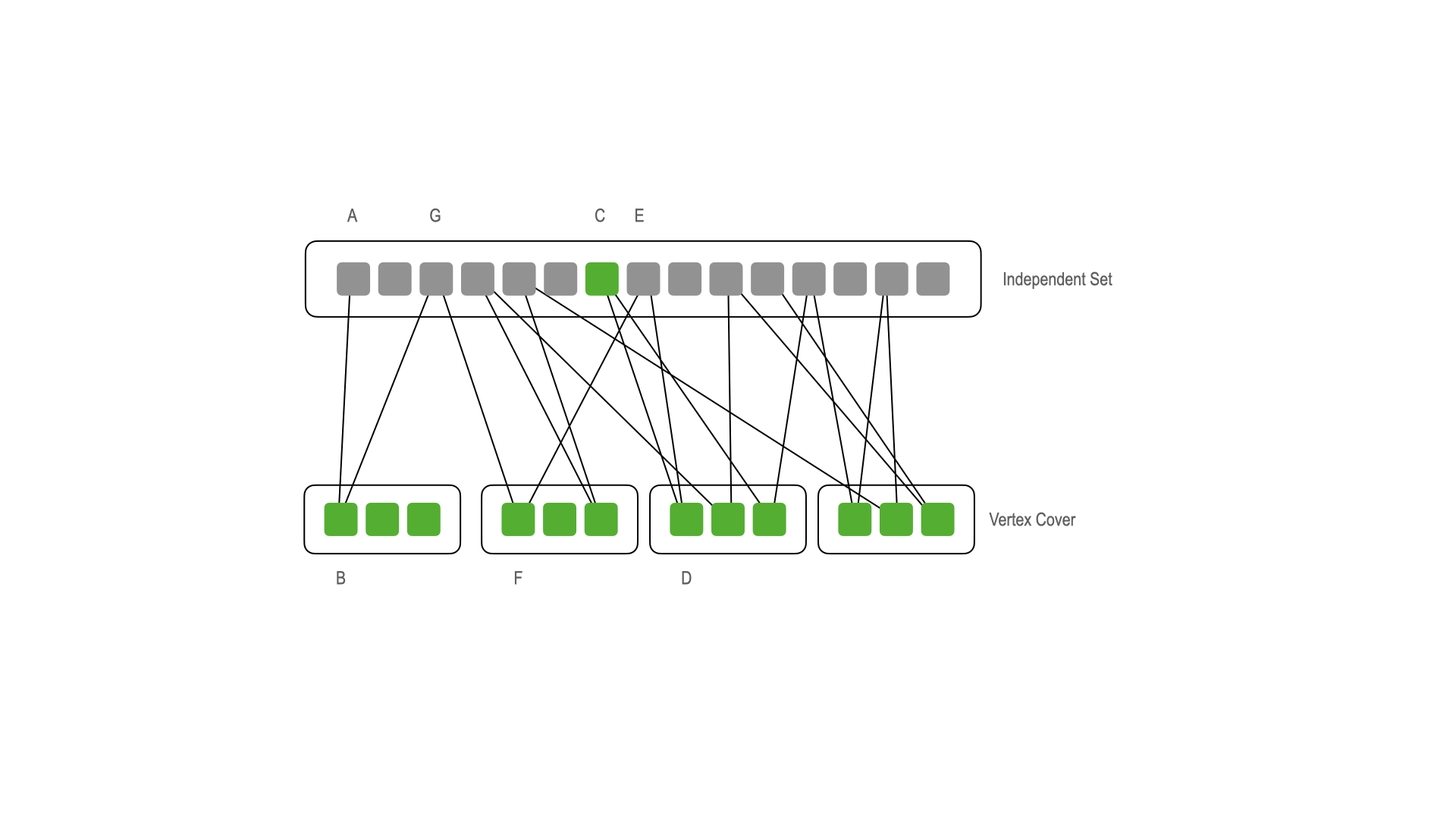}
        \caption{An example graph.}
        \label{fig:y equals x}
    \end{subfigure}
    \hfill
    \begin{subfigure}[b]{0.49\textwidth}
        \centering
        \includegraphics[width=\linewidth]{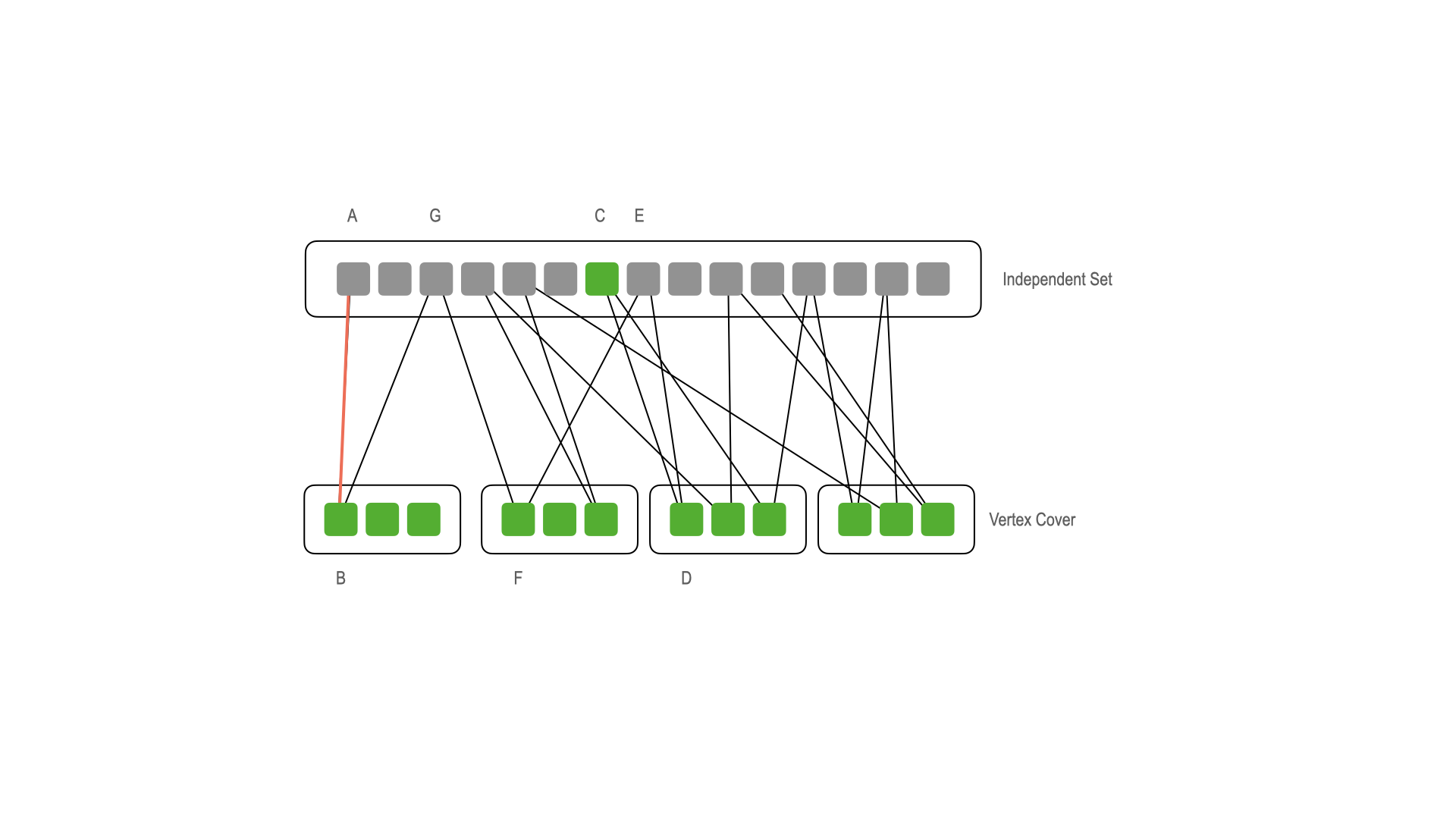}
        \caption{The edge AB is attacked.}
        \label{fig:three sin x}
    \end{subfigure}
    \hfill
    \begin{subfigure}[b]{0.49\textwidth}
        \centering
        \includegraphics[width=\linewidth]{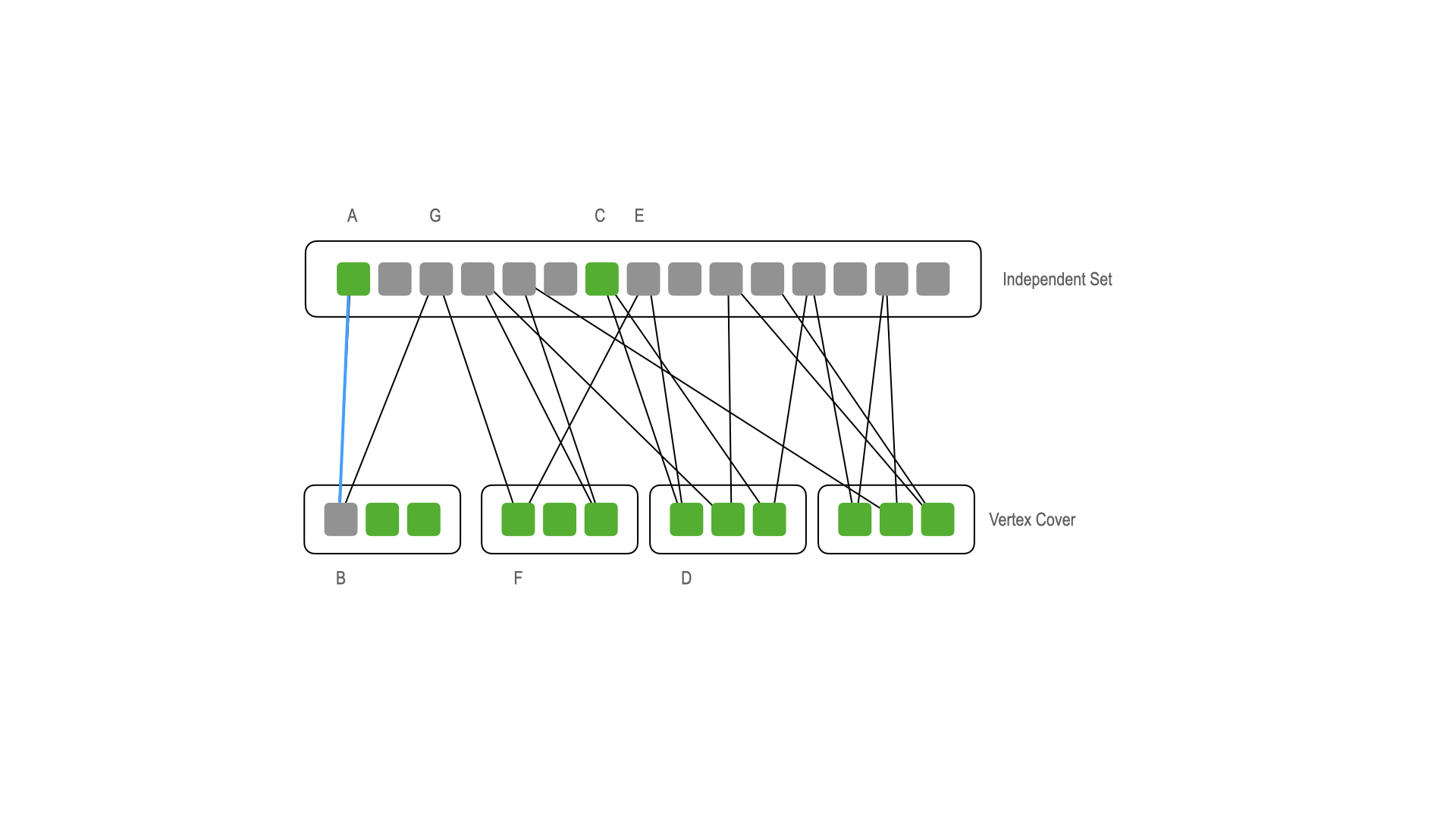}
        \caption{The guard on B moves to A.}
        \label{fig:y equals x}
    \end{subfigure}
    \hfill
    \begin{subfigure}[b]{0.49\textwidth}
        \centering
        \includegraphics[width=\linewidth]{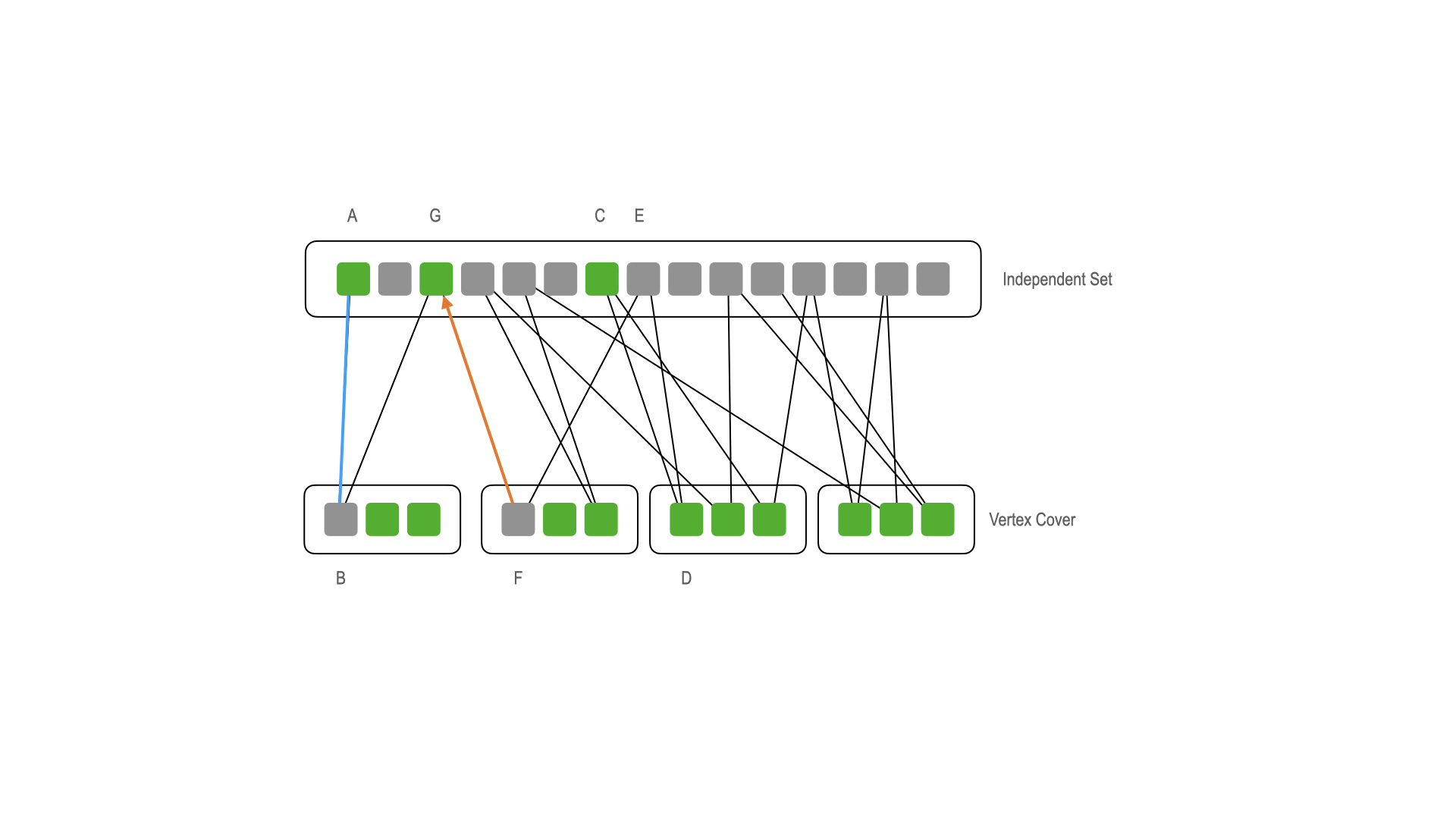}
        \caption{The guard on F moves to G...}
        \label{fig:three sin x}
    \end{subfigure}
    \hfill
    \begin{subfigure}[b]{0.49\textwidth}
        \centering
        \includegraphics[width=\linewidth]{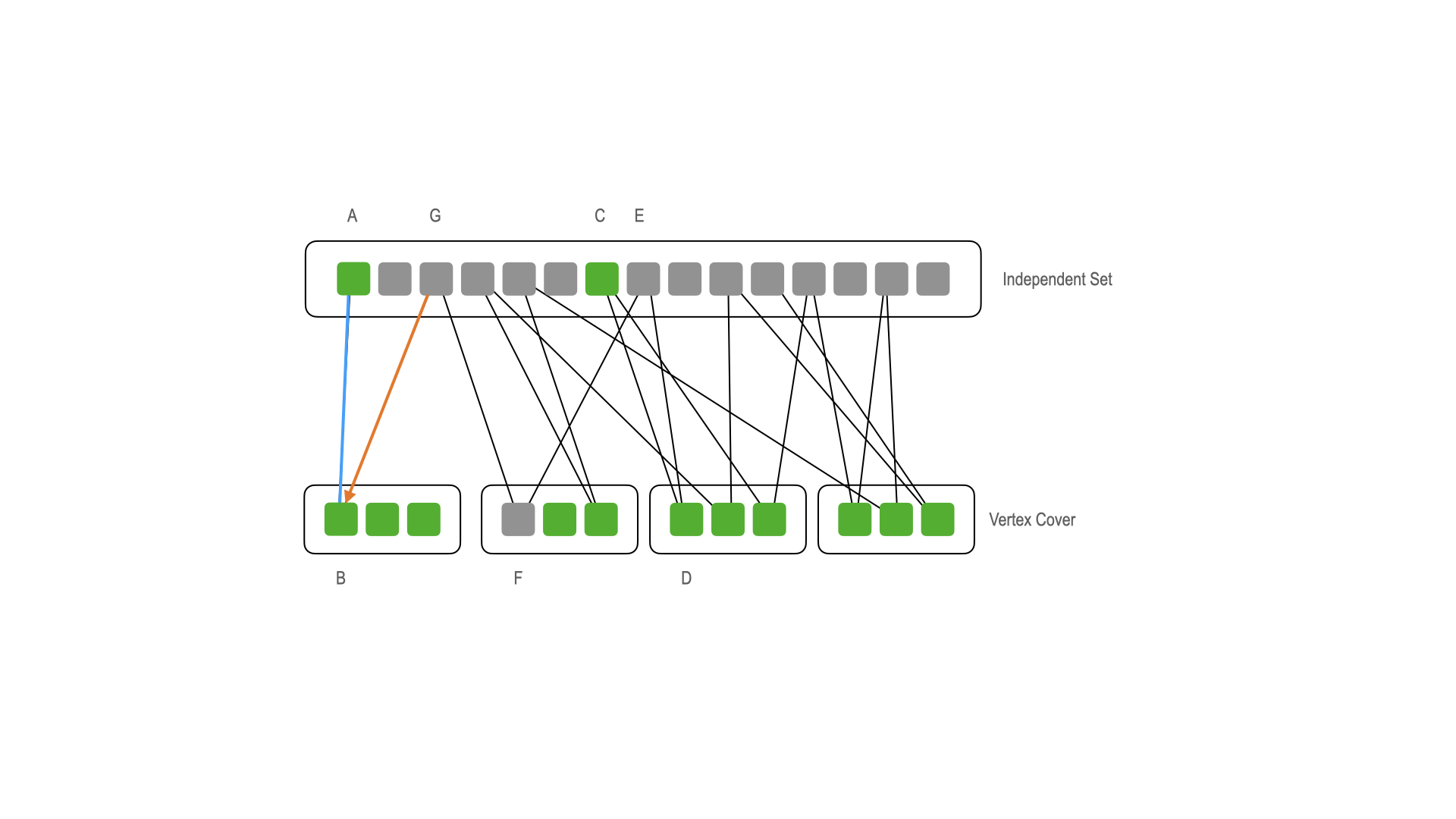}
        \caption{...and then to B.}
        \label{fig:y equals x}
    \end{subfigure}
    \hfill
    \begin{subfigure}[b]{0.49\textwidth}
        \centering
        \includegraphics[width=\linewidth]{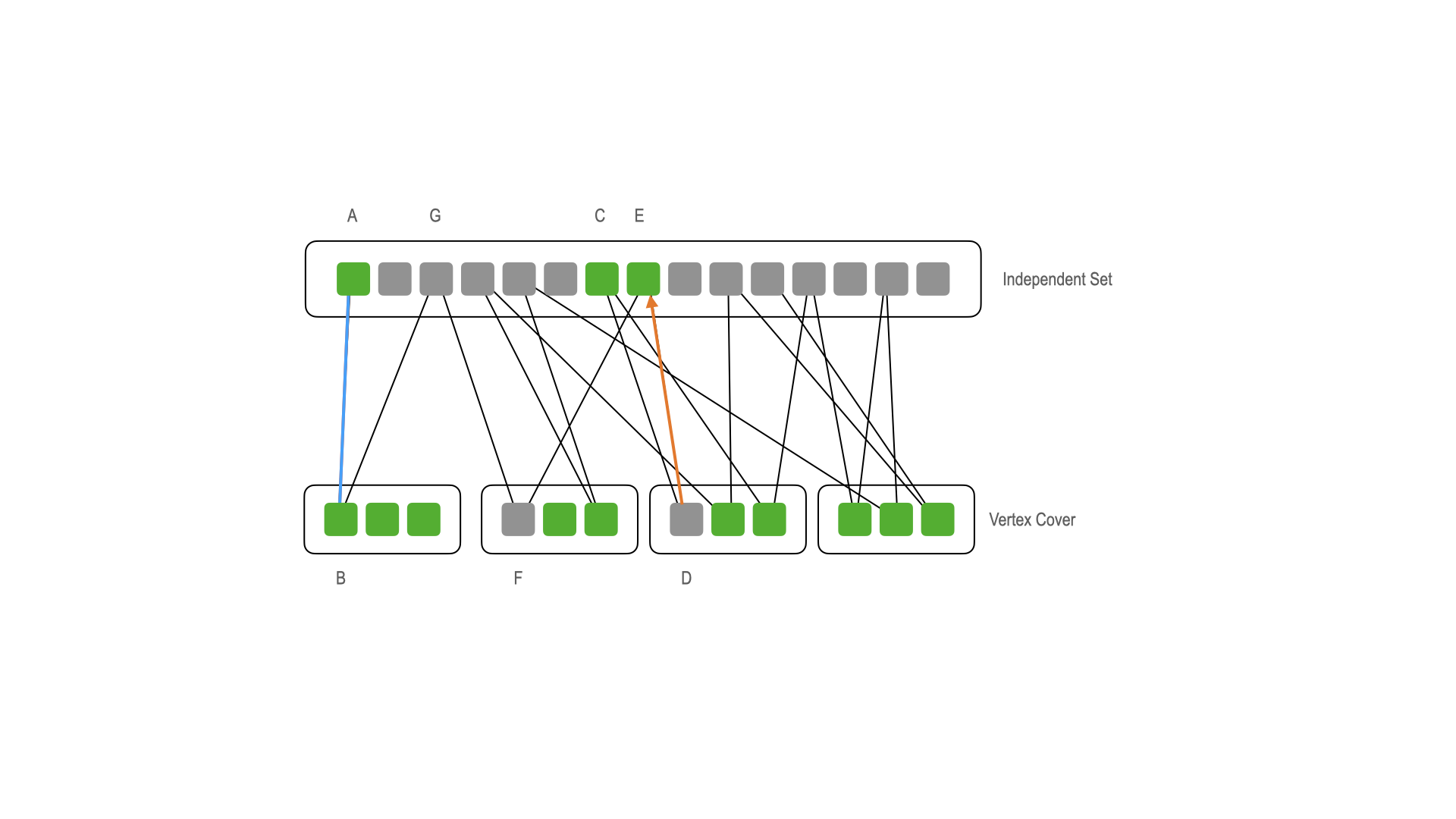}
        \caption{The guard on D moves to E...}
        \label{fig:three sin x}
    \end{subfigure}
    \hfill
    \begin{subfigure}[b]{0.49\textwidth}
        \centering
        \includegraphics[width=\linewidth]{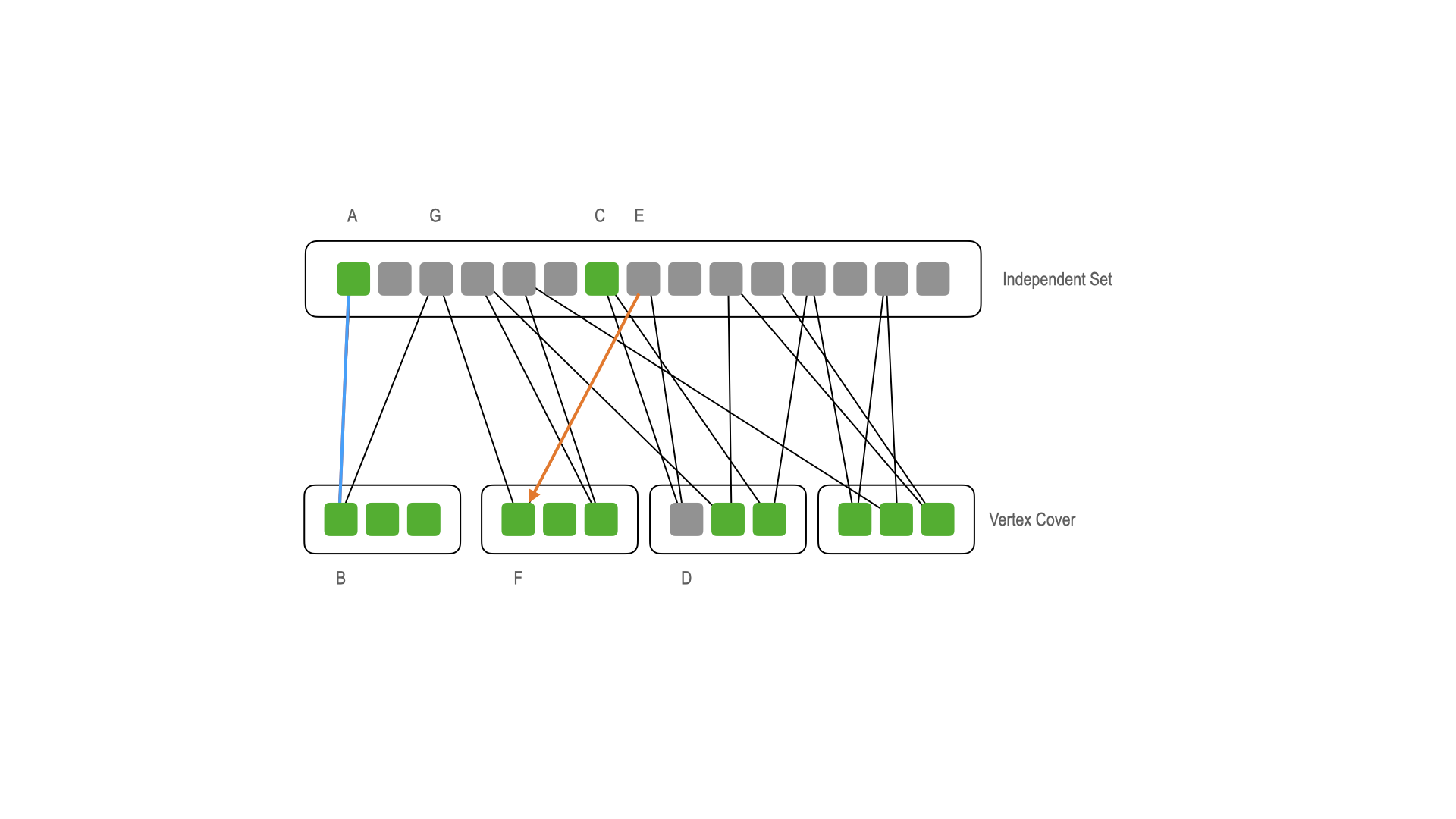}
        \caption{...and then to F.}
        \label{fig:y equals x}
    \end{subfigure}
    \hfill
    \begin{subfigure}[b]{0.49\textwidth}
        \centering
        \includegraphics[width=\linewidth]{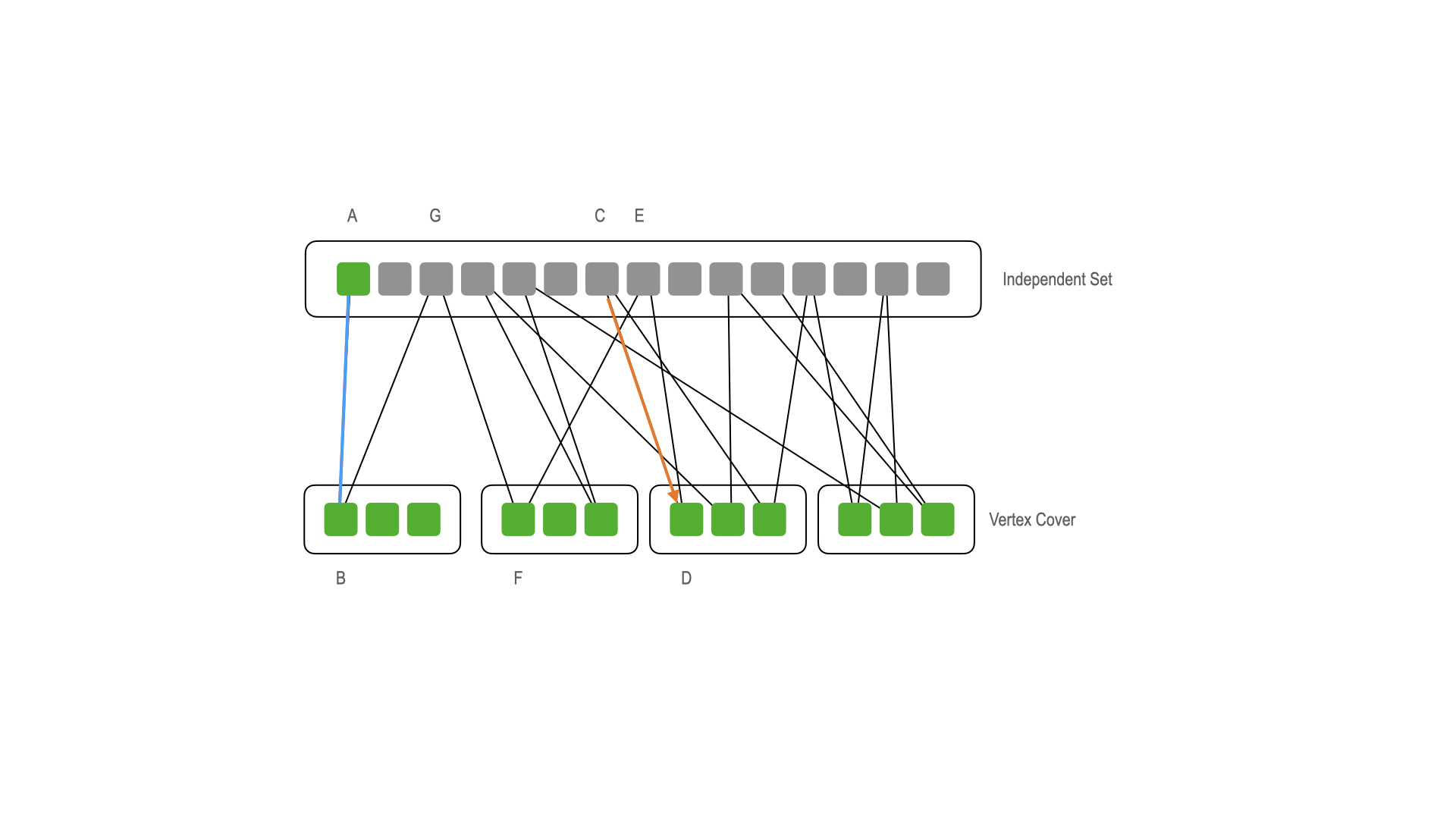}
        \caption{The (originallyextra) guard on C moves to D.}
        \label{fig:three sin x}
    \end{subfigure}
       \caption{Demonstrating a defense with $mvc(G) + 1$ many guards on a graph.}
       \label{fig:defense}
\end{figure}

Let $mvc(G)=k$ and let $C_1,C_2,\ldots,C_q$ be the connected components of the subgraph induced by the vertices of a $k$-sized vertex cover of $G$. Place one guard on each of these vertices of this vertex cover (say $S$) and one guard in the independent set (say on the vertex $w$).

Now suppose some edge is attacked such that both its endpoints have guards, then the guards can exchange places and the configuration is restored.

Suppose some edge $uv$ is attacked such that $u\in S$ and $v\in V(G)\setminus S$. Then the guard on $u$ moves to $v$. Suppose $w$ is adjacent to some vertex which lies in the same connected component of $S$ as $w$, then the guards can each move one step along a path from $w$ to $u$ such that there is a guard on $u$ and all the vertices of $S$ at the end of this movement and now no guard on $w$. But we still have all the vertices of $S$ with one guard each and one guard in the independent set (now on $v$). Thus we have a configuration just like the initial configuration and each strategy can be implemented here accordingly.

Suppose $u$ and any neighbour of $w$ do not lie in the same connected component of $S$. As $G$ is connected, there is a path from $w$ to $u$ in $G$ and this path contains some vertices of $V(G)\setminus S$. Let this path be $w=u_1 u_2 \ldots u_{\ell}=u$. We trigger a movement of guards along this path as follows: If $u_{\ell-1}$ has a guard, this guard moves to $u$, otherwise $u_{\ell-2}$ must have a guard because $S$ is a vertex cover. This guard moves to $u_{\ell-1}$ and then to $u$. Similarly for each $u_i\in S$ where $i\in \{\ell-2,\ell-3,\ldots,2\}$, if $u_{i-1}$ has a guard, the guard on $u_{i-1}$ moves to $u_i$. Otherwise $u_{i-2}$ must have a guard (because $S$ is a vertex cover) which moves to $u_{i-1}$ and then to $u_i$. It can be observed that at the end of this movement, all vertices in $S$ have a guard and $v$ has a guard and $w$ does not have a guard. Therefore, we have the same situation as before where all the vertices in the vertex cover have a guard and one guard outside the vertex cover. Thus we have shown that $mvc(G)+1$ many guards are sufficient to defend $G$ against any sequence of attacks in this new setting.    
\end{proof}

\begin{lemma}
If $G$ is any graph without a degree-$1$ vertex, then $nevc(G)=mvc(G)$. 
\end{lemma}

\begin{proof}
It is sufficient to show that $nevc(G)\leq mvc(G)$, i.e., the defender has a winning strategy using $mvc(G)$ many guards. Also it is sufficient to only look at connected graphs.

Let $mvc(G)=k$ and let $S$ be a minimum sized vertex cover of $G$. The defender places one guard on each vertex of $S$. 
If any edge with both endpoints in $S$ is attacked, then the guards exchange their positions and we are back to the same configuration.

Suppose any edge $uv$ such that $u\in S$ and $v\notin S$ is attacked. The guard on $u$ moves to $v$. Since $v$ is not a degree $1$ vertex, it must have a neighbour $w$ other than $u$. Depending on whether $w$ and $u$ lie in the same connected component in $G[S]$ or not, we can trigger a movement of guards just like the proof of \Cref{sufficient} such that each guard moves at most two steps and there is no guard on $v$ at the end of this movement and there is a guard on $u$. Thus the vertex cover $S$ is restored and we have a winning strategy using $mvc(G)$ many guards.    
\end{proof}

\begin{lemma}\label{check}
For any graph $G$, $nevc(G)=mvc(G)$ if and only if for each degree-$1$ vertex $v$ of $G$, there exists a minimum sized vertex cover $S_v$ of $G$ which contains $v$.
\end{lemma}

\begin{proof}
Let $G$ be a graph such that $mvc(G)=k$ and there exists a degree-$1$ vertex $v$ such that no $k-$sized vertex cover of $G$ contains $v$. Let $u$ be the neighbour of $v$, then any minimum sized vertex cover of $G$ must contain $u$. Therefore, in the initial configuration there must be a guard on $u$ and no guard on $v$. If the attacker attacks the edge $uv$, then the defender must move the guard on $u$ to $v$. Since $v$ has no other neighbour than $u$ and retracing is not allowed, the guard on $v$ cannot move anywhere else. Now since there is no vertex cover of size $k$ which contains $v$, the guards cannot reconfigure to form a vertex cover, no matter how the other guards arrange themselves. Thus $mvc(G)$ many guards are not sufficient to protect $G$ and thus $nevc(G)\neq mvc(G)$.

Now suppose that for every degree-$1$ vertex $v$ of $G$, there exists a minimum sized vertex cover $S_v$ of $G$ which contains $v$. Let the size of the minimum sized vertex cover of $G$ be $k$. We now describe a strategy to defend an attack on $G$ with $k$ guards.
Place each guard on a vertex of a minimum sized vertex cover (say $S$) of $G$. Without loss of generality, we assume that some edge $uv$ such that $u\in S$ and $v \notin S$ is attacked. The guard on $u$ is forced to move to $v$. If $v$ is a vertex with degree $2$ or more, then we mimic the strategy in \Cref{sufficient} to get all the guards back on $S$. If $v$ is a degree-$1$ vertex, we show that it is possible to transfer the guards from $S$ to $S_v$ (where $S_v$ is the minimum sized vertex cover containing $v$). 

Denote $S\cap S_v$ by $P$ and $(V(G)\setminus S \cap S_v)\setminus\{v\}$ by $Q$. Here $P$ are the vertices which need to retain a guard and $Q$ are the vertices which need to gain a guard in order to reconfigure from $S$ to $S_v$. Let $|P|=p$ and $|Q|=q$. Clearly $p+q+1=k$, i.e., $k-p=q+1$. Therefore $|S\setminus P|= |Q\cup\{v\}|$. We show that there is a perfect matching between these two sets. Notice that there cannot be an edge with both endpoints in $S\setminus P$ because these vertices are not in $S_v$ and $S_v$ is a vertex cover. Suppose there is no perfect matching between $S\setminus P$ and $Q\cup\{v\}$, then there exists a set $A\in S\setminus P$ such that $|N(A)\cap (Q\cup\{v\})|< |A|$. But this is not possible as $(S\setminus P)\setminus A\cup N(A)$ will be a vertex cover of $G$ of size less than $k$. Thus there exists a perfect matching between $S\setminus P$ and $Q\cup\{v\}$ which can be used to reconfigure the guards from $S$ to $S_v$.
\end{proof}

\begin{corollary}
   There exists a polynomial time algorithm for finding the \emph{New Eternal Vertex Cover} number on bipartite graphs. 
\end{corollary}

\begin{proof} 
    Suppose we have a bipartite graph $G$ with $mvc(G)=k$. Any minimum sized vertex cover cannot contain both the degree-$1$ vertex and its neighbour. As shown in \Cref{check}, we need to check if for each degree-$1$ vertex $v$ of $G$, whether there exists a vertex cover of size $k$ of $G$ which contains $v$, i.e., whether there exists a vertex cover of size $k-1$ of $G[V(G)\setminus N[v]]$.  
    
    Since the number of degree-$1$ vertices is bounded by the total number of vertices and the subgraph of a bipartite graph is bipartite and vertex cover is polynomial time solvable on bipartite graphs, we have a polynomial time algorithm for checking if $nevc(G)=mvc(G)$ or not if $G$ is bipartite. 
    
    Since $nevc(G)\leq mvc(G)$ for all graphs $G$, if for some bipartite graph $nevc(G)\neq mvc(G)$, then we have $nevc(G)=mvc(G)+1$. Thus we can determine the \emph{New Eternal Vertex Cover} number of a graph in polynomial time.    
    \end{proof}

Note that if we allow the guards to move for an arbitrary number of steps without retracing after each attack, \Cref{check} still holds, i.e., there are families of graphs (for instance, star graphs) such that even allowing the guards to move for an arbitrary number of steps does not make them ``Spartan'' in the new setting. Thus the power of one extra step subsumes the power of any number of additional steps.

\section{Concluding Remarks}

We showed that a natural sufficient condition for when a graph is Spartan (i.e, when $evc(G) = mvc(G)$) is also necessary in the context of bipartite graphs. Motivated by our proof we extend the notion of eternal vertex cover to a variant where the guards are allowed to move more than one step on their turn, and completely characterize the number of guards needed in terms of mvc: indeed, we show that one extra guard suffices, and is needed only if the graph has a degree one vertex that is not contained in any minimum vertex cover. 

It would be interesting to see what happens if the defense can happen in the second step of a guard. We also showed that while the new variant remains computationally hard, unlike the original problem, it is in fact solvable in polynomial time on bipartite graphs. Generalizing our structural results beyond bipartite graphs is an interesting direction for future work.

\bibliography{refs}

\end{document}